\documentclass[11pt]{article}
\usepackage{amssymb}
\usepackage{amsmath, amsthm}

\newtheorem{theorem}{\bf Theorem}[section]

\newtheorem{proposition}{ \bf Proposition}[section]
\newtheorem {definition} {\bf Definition}[section]
\newtheorem {remark} {\bf Remark}[section]
\def\0{{\bf 0}}

\def\A{{\mathbf{A}}}

\def\r{{\bf r}}
\def\p{{\mathbf{p}}}
\def\q{{\mathbf{q}}}
\def\r{{\mathbf{r}}}

\def\R{\mathbb{R}}

\def\u{{\bf u}}
\def\v{{\bf v}}

\def\z{{\mathbf{z}}}
\def\ph{\phantom{-}}

\def\om{\boldsymbol{\omega}}
\def\xii{\boldsymbol{\xi}}
\def\muu{\boldsymbol{\mu}}

\begin{document}

%%%% Article title to be placed here
\title{On the n-body problem in $\mathbb{R}^4$}

\author{%%%% Author details
Tanya Schmah$^{1}$, Cristina Stoica$^{2}$}

\maketitle

%%%%%%%%% Insert author address here
%\address{$^{1}$Department of Mathematics and Statistics, University of Ottawa, Ottawa, K1N 6N5 Canada\\
%$^{2}$Department of Mathematics, Wilfrid Laurier University, Waterloo, N2L 3C5 Canada}

%%%% Subject entries to be placed here %%%%
%\subject{applied mathematics, mathematical physics, geometry, dynamics}

%%%% Keyword entries to be placed here %%%%
%\keywords{n-body problem in $\mathbb{R}^4$, general potential, SO(4)-symmetry, relative equilibria, regular n-gon,  reduction, three-body,  stability}

%%%% Insert corresponding author and its email address}
%\corres{Cristina Stoica\\
%\email{cstoica@wlu.ca}}

%%%% Abstract text to be placed here %%%%%%%%%%%%
\begin{abstract}
Using geometric mechanics methods, we examine aspects of the dynamics of  n mass points in $\mathbb{R}^4$ with a general pairwise potential. We investigate the central force problem, set up the  n-body problem  and discuss certain properties of relative equilibria.  We describe regular n-gons  in $\mathbb{R}^4$ and when the masses are equal,
we determine the invariant manifold of motions with regular n-gon configurations. 
In the case n=3 we reduce the dynamics to a six degrees of freedom system and we show that for generic potentials and momenta, relative equilibria  with equilateral configuration are unstable. 
\end{abstract}

%%%%%%%%%%%%%%%%%%%%%%%%%%%

\tableofcontents

\section{Introduction}

In the last decades a good body of work has been devoted to the study of various generalizations of the classical $n$-body problem. Thus we find studies  of $n$ mass points with modified potentials (\cite{APS14, BVC17,DMS00, MG81, St00}), or on configuration spaces with a non-Euclidean structure
 (\cite{Shc06, Dia14, MS15, DS18, St18}), or in higher dimensional Euclidean spaces (\cite{PJ80, PJ81a, PJ81b, OV06, AC98, Ch13, Mo14, AMP16}). 

In this paper we investigate aspects of  the  dynamics in $\mathbb{R}^4$ of $n$ mass points with binary interaction depending on  distance only.    Working from a  geometric mechanics perspective, we identify  a variety of  interesting properties and  retrieve some known results.

As a preamble, we review  the geometry of the $SO(4)$ action on $\mathbb{R}^4$.  It is known that any matrix in $SO(4)$ can be expressed as the product of two planar rotations, each in
its own invariant plane, with the two planes mutually orthogonal. 
Thus any element of $SO(4)$ is conjugate to one in the \textbf{double planar rotation group},
defined in an $Oxyzw$ cartesian system of coordinates as  
  $SO(2)_{xy}\times SO(2)_{zw}$, the subgroup of $SO(4)$  leaving 
  the \textbf{principal planes} $Oxy$ and $Ozw$ invariant.
This is a normal form for rotations, and can 
be understood as a generalization of Euler's rotation theorem. It reduces many questions   about $SO(4)$ symmetries to questions about symmetries with respect to the action
of the double planar rotation group.
 
We choose to focus on the double planar rotation group rather than the full symmetry group $SO(4)$ since it simplifies our study (in particular, reduction simplifies  greatly) and we still
obtain a rich variety of dynamical features. 
In particular, any relative equilibrium of the full $SO(4)$ action is conjugate
to a relative equilibrium with respect to the action of the double planar rotation group;
see Remark \ref{REnormalform}.

We start by  taking  a close look at   the central force problem. We find  that on each of the principal planes,
the projection of the motions obeys the area laws. We also identify various invariant manifolds, in particular that of \textbf{proportional motions} for which the polar radii of the projections of onto the principal planes have a constant ratio.

We briefly consider the Kepler problem for the $\mathbb{R}^3$-Newtonian  potential $V(r)= -k/r$,  $k>0$ ($r$ being the distance between two mass points).
Note that this potential 
 is not the solution of the Laplace equation in $\mathbb{R}^4$, and so is not the standard model
 in mathematical physics. However the induced dynamics is interesting and was and is considered for theoretical significance (see, for instance, \cite{PJ80, PJ81a, PJ81b, AC98, OV06, Ch13}).  We observe that the dynamics under the $\mathbb{R}^4$  gravitational potential, that is $V(r)= -k/r^2$,  $k>0$, also known as the Jacobi potential (see \cite{Al15} and references within), presents intriguing  degeneracies and  we defer its detailed study  for future projects. 
 
 In  the Newtonian  potential $V(r)=-1/r$, we prove that in the Kepler problem collisions are possible if and only if the angular momentum is zero. We also note that this conclusion is  true for any homogeneous law, except  the aforementioned attractive inverse square Jacobi potential.

Recall that in the classical  ($\mathbb{R}^2$  or $\mathbb{R}^3$)  $n$-body problem, a homographic solution  is one such that the configurations formed by the masses  are similar modulo a rotation  and a scalar multiplication. Two particular cases are usually highlighted: the  relative equilibria (RE), for which  the scalar multiplication is the identity; and homothetic solutions, which display no rotation. Along a homographic solution the bodies form a \textbf{central configuration}.  In the classical Newtonian case, this  is equivalent to 
having the configuration as 
 a   critical point of the potential among configurations with a given moment of inertia.  Perhaps the most important property of central configuration concerns total collision: when the bodies are released from the central configuration with zero initial velocity, they end in total collapse.

% (and so they are critical points of the potential with a given moment of inertia), but this not true in $\mathbb{R}^4$.

\iffalse
In $\mathbb{R}^4$ the same definitions apply but with the rotation group given by $SO(4)$. In this paper, however, we define RE via the general definition for Lie-symmetric ODE systems: a solution is a RE if it is also a one-parameter group orbit.
%; its configuration $\q_0$ is  called to be a \textbf{base point}. 
We  note that in the  planar $n$-body problem of celestial mechanics  the RE configurations  are in fact the central configurations, but this not true in $\mathbb{R}^4$.
\fi

The Newtonian $n$-body problem  in higher dimensional Euclidean spaces was previously studied by Palmore  (\cite{PJ80, PJ81a, PJ81b}) and Albouy and Chenciner (\cite{AC98, Ch13}, but see also \cite{Mo14}). 
%Here we allow the potential to be a general rotational invariant function. 
Palmore focuses on homographic motions and RE, finding, amongst other results, that there are  RE solutions which are not central configurations, and  that the configurations  of the homothetic solutions  must always be  central. Albouy and Chenciner present an extensive study in  Euclidean spaces of any dimension. They prove that a  necessary condition for the 
existence of non-homothetic homographic motion is that the motion takes place in an even-dimensional space. 
Further, two cases are possible: either the configuration is central (the latter being defined as a critical point of the potential among configurations with a given moment of inertia) and the space where the motion takes place is endowed with a hermitian structure; or it is \textit{balanced}, that is, we have a critical point of the potential among configurations with a given inertia spectrum, and the motion is a new type, quasi-periodic, of relative equilibrium.

In $\mathbb{R}^4$ we define RE via the general definition for Lie-symmetric ODE systems: a solution is a RE if it is also a one-parameter group orbit.
We show that  balanced configurations are exactly the  configurations of RE, and 
that these are central configurations if and only if the components of the angular velocity on the two principal planes are equal. We further investigate collinear RE and find that either they have configurations  in one of the principal planes, or their angular velocity  components are equal.

%When the  masses are equal, the  $n$-body problem is amenable to methods that use the finite symmetries given by the permutation of the points' labels.  Thus, using Discrete Reduction \cite{Ma92} and Palais' Principle of Symmetric Criticality \cite{Pal79}, we retrieve the invariant manifolds of homographic motions. On these, the  bodies maintain  a regular $n$-gon configuration at all times and the dynamics reduce to the central force problem. We observe that while the projection of the $n$-gon on each of the principal planes is a regular $n$-gon as well, the rotations in each plane might have different angular velocities. 
%The dynamics on the homographic invariant manifold coincide to that of the central force problem.

% Note that this result is valid for any rotational-invariant potential.

When the masses are equal, due to the finite symmetries, we are able to detect  low dimensional invariant manifolds using the method of \textit{Discrete Reduction} \cite{Ma92} and Palais' Principle of Symmetric Criticality \cite{Pal79}. These invariant manifolds  consist of equilibria and relative equilibria 
with configurations that are ``regular" $n$-gons  in $\R^4$.
 By a ``regular'' $n$-gon we do not simply mean that the side lengths are equal (which would include e.g. all rhombuses), nor do we
wish to consider only planar shapes. Instead, we follow Coxeter \cite{Cox36}: 
``A polygon (which may be skew) is said to be regular if it
possesses a symmetry which cyclically permutes the vertices (and therefore
also the sides) of the polygon.'' 

We %discuss the definition above as applied to our context and 
classify regular $n$-gons in $\mathbb{R}^4$ in Proposition \ref{prop:reg}.  A \textbf{planar} $n$-gon lies in a 2-dimensional plane in 
$\R^4$; either it lies in a principal plane, or it projects to similar $n$-gons in each of the principal planes.
These two $n$-gons are \textit{synchronised}, meaning that there exists a labelling of the points such that each of the projected $n$-gons is convex.
For a \textbf{nonplanar (skew)} $n$-gon, the projections onto the two principal planes are not similar. 
 The nonplanar regular $n$-gons are of two types. In type (I) both projections onto the principal planes 
 are regular $n$-gons, but they are not synchronised in the sense defined above. In type (II),
at least one of the two projections onto a principal plane has fewer than $n$ sides;
denoting by $b_1$ and $b_2$ the number of sides of the two projections,
$n$ is the lowest common multiplier (lcm) of $b_1$ and $b_2$.
%, where $b^{(1)}_j$ and $b^{2}_j$ are the number of the sides of the projections onto the two principal planes.
Regular planar polygons may have any number of vertices.
The smallest nonplanar polygon of type (I) has $n=5$ vertices, with one projection convex and the other a pentagram.
The smallest nonplanar polygon of type (II) has $n=4$ vertices, with one projection a square and the other a digon (i.e. having $2$ vertices). 
The next smallest nonplanar polygon of type II has $n=6$ vertices, with one projection a triangle and the other a digon. 
The three examples are the only nonplanar polygons with fewer than $7$ vertices.

We also specifically address  the three body  problem. Taking advantage of the Abelian structure of symmetry group $SO(2)_{xy} \times SO(2)_{zw}$, we are able to reduce   the system from 12 to  6 degrees of freedom using directly the integrals of motion. Then we write  the RE conditions and deduce  that for the $\mathbb{R}^3$-Newtonian $V(r)=-1/r$ law, any non-collinear RE must be isosceles; this condition was observed  previously in \cite{AC98}. 
%We also note that this is true for any  homogeneous potential that is  directly dependent on the product of the  masses.  

When the three masses are equal  we study the stability of the equilateral triangle RE (in the case of attractive interactions).
We find that such RE  exist  for any attractive potential. 
%(A note that in the case of the  Jacobi $V(r)=-1/r^2$ law, these RE are not isolated and their existence is insured by an additional condition over momenta is given at the end of the paper.)
%Calculating the appropriate matrix blocks, w
We determine that equilateral RE generically are unstable. Specifically, provided a certain sub-block of the  the Hessian of the amended potential has non-zero determinant, the matrix linearization at the RE displays a nilpotent component responsible for a ``drift" component in the dynamics.

The paper is organized as follows: in Section \ref{Geo} we review the geometry of the $SO(4)$ action on $\mathbb{R}^4$, including the orbit types and the calculation of the momentum map. In Section  \ref{sect: central} we study the central force problem, detect certain invariant manifolds and discuss briefly the Kepler problem and collisions. In the next section we introduce the general $n$-body problem.  We examine RE, balanced and central configurations and collinear RE. In Section \ref{sect:DR} we explore regular $n$-gons in $\mathbb{R}^4$, % following Coxeter's definition, 
and apply the Discrete Reduction method to detect the invariant manifolds of regular $n$-gons.  In Section  6 we set up the $3$-body problem and reduce the dynamics to a system of 6 degrees of freedom. In the case of equal masses, we study  the stability of equilateral configurations REs and find that with the class of attractive potentials, these are generically unstable.
We conclude with some remarks and open questions.

%{sect:DR}

%{sect:3-body}

\section{Geometric set-up}\label{Geo}

The Lie group $SO(4)$ has a standard action on $\mathbb{R}^4$ and further acts on $T^*\mathbb{R}^4 \simeq \mathbb{R}^4 \times \mathbb{R}^4$ via the co-tangent lift $(R\,, (\q, \p))) =  (R\q, R\p)$ for any $R\in SO(4)$ and $(\q,\p) \in T^*\mathbb{R}^4.$
%We refer to elements of $SO(4)$ as ``rotations''.
The diagonal action of $SO(4)$ on configurations of n bodies in $\R^{4n}$
is $R(\q_1,\dots,\q^n) = (R\q_1,\dots,R\q^n)$; this too has a co-tangent lifted action 
on $T^*\R^{4n}$.
We consider $n$-body problems with $SO(4)$ symmetry, 
arising from $SO(4)$-invariant potentials $U:\mathbb{R}^4 \to \mathbb{R}$.
Note that any $U$ that depends only on pairwise distances has this property.

Given any $R \in SO(4)$ it is always possible to find an orthogonal change of coordinates
that transforms $R$ to
\begin{equation}
%S_{\left(\theta_1, \theta_2\right)}
S := \left[\begin{array}{cccc}
\cos \theta_1  & -\sin \theta_1 &0 & 0\\
\sin \theta_1  & \ph \cos \theta_1 &0 & 0\\
0&0& \cos \theta_2  & -\sin \theta_2\\
0&0& \sin \theta_2  & \ph \cos \theta_2 \\
\end{array}
\right]
,
\label{E:rot}
\end{equation}
for some $\theta_1, \theta_2 \in \R$.
In other words, there exists a $Q \in SO(4)$ such that
$R = Q S Q^t$, with $S$ in the above form
\cite[Ex. 1.15]{Hall15}.
% Direct way: spectral theorem, then take sums and differences of eigenvectors, as in
% \cite{Hall15} exercise, solved in MAT4144. 
% Indirect, fancy schmanzy way:
%Let $T = SO(2)_{xy} \times SO(2)_{zw}$. This is a maximal torus in $SO(4)$ \cite{Still08}.
%https://en.wikipedia.org/wiki/Rotations_in_4-dimensional_Euclidean_space
%Since $SO(4)$ is a compact, connected Lie group, every element of G is conjugate to an element of T.
The form of $S$ in \eqref{E:rot} is a \textit{normal form} for $R$. 
%Recall that eigenvalues are invariant with respect to similarity, so 
Note that $S$ and $R$ have the same eigenvalues.
The normal form is uniquely determined by these eigenvalues,
$\cos \theta_j \pm i \sin \theta_j$ for $j=1,2$,
except for the possible exchange of the two diagonal blocks (if 
$\theta_1$ and $\theta_2$ are unequal).
Exchange of the two blocks in $S$ can be accomplished by the 
special orthogonal transformation $\tau$ defined by 
$\tau(x,y,z,w) = (z,w,x,y)$, with matrix
\begin{align}\label{E:tau}
\left[\begin{array}{cccc}
0  & 0 &1 & 0\\
0  & 0 &0 & 1\\
1 &0& 0  & 0\\
0&1  & 0  & 0 \\
\end{array}
\right].
\end{align}
%Every conjugacy classes of rotation matrices
%contains either one rotation in normal form (if $\theta_1 = \theta_2$) or two (otherwise).

From the normal form, it is clear that $R$ has two mutually orthogonal invariant planes,
and $R$ consists of a rotation in each of these planes. 
%``to two rotations that rotate two planes". 
%two planar rotations.
This is  analogous to Euler's rotation theorem:  in $\mathbb{R}^3$, any rotation is equivalent to a single rotation about some axis that runs through the origin.

Note that in the 3D case, the axis of rotation is uniquely defined if and only if $R$ is not the identity.
In the 4D case, we have a similar situation. Indeed, 
each invariant plane corresponds to a pair of 
complex conjugate eigenvalues of $S$,
$\cos \theta_j \pm i \sin \theta_j$, for $j=1,2$.
Generically these are $4$ distinct eigenvalues in $2$ distinct pairs,
each pair corresponding to an invariant subspace.
It is easily checked that the only exceptions are: 
(i) $\theta_1 = \theta_2 = 0$, i.e. $R = S = Id$; and
(ii) %$\cos \theta_1 = \cos \theta_2 = 0$, i.e. 
$\theta_1 = \pm  \theta_2 = \pm \frac{\pi}{2}$, in which case there are two double eigenvalues
$\pm i$.
In both cases, there exist orthogonal changes of coordinates that move the
invariant subspaces while leaving the form of $S$
unchanged, for example if $\theta_1 = \theta_2 = \pm \frac{\pi}{2}$ then, for any $\alpha \in \R$,
\begin{align}\label{E:noninvplane}
P = 
\left[\begin{array}{cccc}
\cos \alpha  & 0 &-\sin \alpha & 0\\
0  & \cos \alpha &0 & -\sin \alpha\\
\sin \alpha &0& \cos \alpha  & 0\\
0&\sin \alpha & 0  & \cos \alpha \\
\end{array}
\right]
\quad \implies P S P^t = S.
\end{align}

In summary, we have shown the following.

\begin{proposition}\label{principalR}
(i) Any $R \in SO(4)$ is orthogonally similar to the normal form given in \eqref{E:rot},
which is unique except for the possible exchange of the two diagonal blocks.
Two matrices $R_1, R_2 \in SO(4)$
have the same normal form 
(up to a possible exchange of diagonal blocks)
if and only if they are orthogonally similar, i.e. $R_2 = QR_1 Q^t$ for some $Q \in SO(4)$.\\
(ii) The invariant planes of an $R\in SO(4)$ are uniquely defined
unless $R = Id$ or each of the two blocks in \eqref{E:rot}
equals either 
$\left[\begin{array}{cc}
0  & -1\\
1  & 0
\end{array} \right]$
or $\left[\begin{array}{cc}
0  & 1\\
-1  & 0
\end{array} \right]$, 
i.e. $\theta_1 = \pm  \theta_2 = \pm \frac{\pi}{2}$,
in which cases there are an infinite number of invariant planes.
\end{proposition}

The Lie algebra of $SO(4)$ is $so(4)$, the set of infinitesimal rotations in $\R^4$.
It consists of all skew-symmetric matrices.
%https://en.wikipedia.org/wiki/Skew-symmetric_matrix#Infinitesimal_rotations
There is a normal form for $so(4)$ closely related to the one
above for $SO(4)$:
for any $\xi \in so(4)$, there exists a $Q\in SO(4)$ such that $\xi = Q \hat \om Q^t$, with
\begin{align}\label{E:infrot}
\hat \om =\left[
\begin{array}{cccc}
0 & -\omega_1 &0 & 0\\
\omega_1 & 0& 0 & 0 \\
0 & 0& 0& -\omega_2\\
0 & 0& \omega_2 & 0
\end{array}
\right],
\end{align}
for some $\omega_1, \omega_2 \in \mathbb{R}$\,.
This equation also defines the ``hat'' notation $\hat \om$ for a vector 
$\om = (\omega_1, \omega_2) \in \mathbb{R}^2.$
%From the above discussion, taking $\theta_1$ and $\theta_2$ both equal to $\frac{\pi}{2}$,
%and $\omega_j = \sin \theta_j$ for $j=1,2$, we see that the normal form $\hat \om$
%does \textit{not} have uniquely-defined invariant planes.
This normal form is well-known, however since we are unaware of 
a reference for it, we give a brief proof in (i) below that it is a consequence of Proposition \ref{principalR}.

\begin{proposition} \label{principalom}
(i) Any $\xi \in so(4)$ is orthogonally similar to the normal form given in \eqref{E:infrot}, 
which is unique except for the possible exchange of the two diagonal blocks.
Two matrices $\xi_1, \xi_2 \in so(4)$
have the same normal form 
(up to a possible exchange of diagonal blocks)
if and only if they are orthogonally similar, i.e. $\xi_2 = Q\xi_1 Q^t$ for some $Q \in SO(4)$.\\
(ii) The invariant planes of $\hat\om$ are uniquely defined
unless $\om = (\omega, \omega)$, i.e. $\omega_1 = \omega_2$, in which case 
there is an infinite family of invariant planes.
\end{proposition}

\noindent
Proof: (i) Since $\exp(\xi) \in SO(4)$,  
Proposition \ref{principalR} implies that $\exp(\xi) = Q S Q^t$
for some orthogonal $Q$ and some $S$ in normal form \eqref{E:rot} for some $\theta_1, \theta_2$.
Let $\om = (\omega_1, \omega_2) = (\theta_1, \theta_2)$ and $\hat\om$ as in \eqref{E:infrot}. Then 
since $S = \exp(\hat\om)$ and
\[
\exp(\xi) = Q S Q^t = Q \exp(\hat\om) Q^t = \exp(Q \hat\om Q^t).
\]
It follows that $\exp(t \xi) = \exp(Q \widehat{t \om} Q^t)$ for all $t$.
Since $\exp$ is locally invertible in a neighbourhood of zero \cite{Hall15},
the above calculation, applied to a small enough $t$, implies that
$\xi = Q \hat \om Q^t,$
so $\xi$ has the normal form \eqref{E:infrot}.
The normal form is completely determined by the eigenvalues, 
which are $\pm i \omega_1$ and $\pm i \omega_2$.
Thus two unequal matrices $\hat \om$ and $\hat \om'$ are 
are orthogonally similar if and only if they differ
by an exchange of diagonal blocks, i.e. $\hat \om = \tau \hat \om \tau^t$
for the orthogonal matrix $\tau$ in \eqref{E:noninvplane}.
\\
(ii): The eigenvalues of $\hat \om$ are $\pm i \omega_1, \pm i \omega_2$.
If these are distinct then each pair uniquely defines an invariant plane. 
If $\omega_1 = \omega_2$, then \eqref{E:noninvplane} gives a change of coordinates
leaving the form of $\hat \om$ unchanged, showing that there is an infinite family
of invariant planes.
$\square$

\begin{remark}\label{torustheorem}
The preceding two propositions concern single elements of $SO(4)$ or $so(4)$.
There is a related result for subgroups of $SO(4)$:
any compact abelian subgroup is conjugate to a subgroup of
$SO_{xy} \times SO_{zw}$. This is a consequence of the Torus Theorem
for Lie groups \cite{Hall15}, and the fact that 
$SO_{xy} \times SO_{zw}$ is a maximal torus in $SO(4)$\cite{Still08}.
%\cite{Hall15}.
%https://en.wikipedia.org/wiki/Rotations_in_4-dimensional_Euclidean_space
%Since $SO(4)$ is a compact, connected Lie group, every element of G is conjugate to an element of T.
\end{remark}

Because of the properties just stated, many questions about $SO(4)$ symmetries
can be reduced to questions about symmetries with respect to 
$SO(2)_{xy} \times SO(2)_{zw}$.
With the exception of Section 4b, throughout the paper  we consider 
the symmetries associated with this group, which we call the
the \textbf{double planar rotation group}.
% is $SO(2)_{xy} \times SO(2)_{zw}$.
The \textbf{principal planes} are $Oxy$ and $Ozw$.
Note that whenever an element $R \in SO(2)_{xy} \times SO(2)_{zw}$
has uniquely defined invariant planes, they equal the principal planes;
this is also true for any $\xi \in so(2)_{xy} \times so(2)_{zw}$.

\smallskip

Any element $\xi$ of the Lie algebra $so(4)$ determines a unique one parameter group
$R(t) := \exp t\xi$, and the derivative of this path at $t=0$ is $\xi$. 
Given a ``base point'' $\q \in \R^4$, it follows that $\xi$ determines a path 
in $\R^4$ given by $\left(\exp t \xi \right) \q$, 
and the derivative of this path at $t=0$ is $\xi \q$ (matrix product),
which is called the \textbf{infinitesimal action} of $\xi$ on $\q$.
Note that if $\xi = \hat \om$ as defined in \eqref{E:infrot}, and
$\q=(x, y, z, w) \in \mathbb{R}^4,$
then
\begin{equation}
 \hat\om \q = 
\left[
\begin{array}{c}
-\omega_1 y\\
\ph \omega_1 x\\
-\omega_2 w\\
\ph\omega_2 z\\
\end{array}
\right].
\label{inf-gen}
\end{equation}
If $SO(4)$, or its subgroup $T=SO(2)_{xy} \times SO(2)_{zw}$, is acting
diagonally on $\R^{4n}$, 
then the infinitesimal action of $\xi$ on $\q = \left(\q_1, \dots, \q_n\right)$ is
$\xi \q = \left(\xi \q_1, \dots, \xi \q_n\right)$.

\smallskip

For any group $G$ acting on a space $X$, 
the \textbf{isotropy subgroup} of any $x \in X$ is
the set of all group elements that leave $x$ fixed, i.e.,
%The isotropy subgroup $G_\q$ of a point $\q \in \mathbb{R}^4$ is the subgroup of $G$ consisting of those transformations fixing $\q$, that is 
%
\begin{equation}
G_x:=\{g \in G\,|\, gx=x\}.
\end{equation}
It is easily seen that if $y=gx$ then $G_y= g G_x g^{-1}$, that is the isotropy subgroups of two points in the same group orbit are conjugate. Thus to each orbit is associated a conjugacy class of subgroups of G, called the \textbf{orbit type} of the orbit. 

\smallskip 
In $SO(4)$, conjugacy is the same as orthogonal similarity. 
In any abelian group, such as the double planar rotation group,
conjugacy leaves every subgroup invariant, i.e. $g G_x g^{-1} = G_x$,
so each orbit type contains only one subgroup.

% NO! orbit type might contain SO(3).
\iffalse
To compute orbit types for the standard action of $SO(4)$ on $\R^4$, it suffices to consider elements  of the double planar rotation group,
$T=SO(2)_{xy} \times SO(2)_{zw}$,
since all elements of $SO(4)$ are
conjugate to one in this smaller group.
Specifically,
\begin{proposition}
Two points $\q_1,\q_2 \in \R^4$ have the same orbit type with respect to $SO(4)$ if and only if 
they have the same orbit type with respect to $T$.
This occurs if and only if either
$T_{\q_1} = T_{\q_2}$, or $T_{\q_1} = \tau T_{\q_2} \tau$, where $\tau$ is the 
matrix in \eqref{E:tau} that exhanges diagonal blocks.
\end{proposition}

\noindent
Proof: 
%Every orbit type is an equivalence class with respect to conjugacy within $SO(4)$,
%i.e. with respect to orthogonal similarity. 
It follows from Proposition \ref{principalR} that the conjugacy class of $G_{\q}$
the orbit of every element of $SO(4)$ intersects $G$ in either a unique $A$ or a
pair of elements $\{A, \tau A \tau\}$. Therefore every orbit type intersects $G$ 
since $\q_1$ and $\q_2$ have the same orbit type if and only $G$...?
\fi

\begin{proposition} The orbit types 
for the standard action of 
$G:=SO(2)_{xy} \times SO(2)_{zw}$ on $\R^4$ are:
\begin{enumerate}
\item  $\q=(0,0,0,0)$ has $T_\q=T$;

\item $\q=(x,y,0,0)$ with $x,y\neq 0$  has $G_\q=SO(2)_{zw}$ (rotations of the $Ozw$ plane);

\item $\q=(0,0,z,w)$ with $z,w\neq 0$ then $G_\q=SO(2)_{xy}=$ rotations of the $Oxy$ plane.

\item  If $\q=(x,y,z,w)$ with $(x,y) \ne (0,0)$ and $(z,w) \ne (0,0)$ then $G_\q=\{\mathbb{I}_4\}$
(where $\mathbb{I}_4$ is the identity matrix).

\end{enumerate}
\end{proposition}

\begin{remark}
In $SO(4)$, the two isotropy subgroups $SO(2)_{xy}$ and $SO(2)_{zw}$ are conjugate
by the matrix $\tau$ in 
\eqref{E:tau} that exchanges blocks, so these are in the same 
$SO(4)$ orbit type. In fact, 
by Remark \ref{torustheorem}, all compact abelian subgroups of $SO(4)$ 
of a given dimension are 
conjugate to each other, i.e. all copies of $SO(2)$ are conjugate to each other,
and all copies of $SO(2) \times SO(2)$ are conjugate to each other.
%but in $T$ no two of them are, since $T$ is abelian.
In $SO(4)$ there are also non-abelian isotropy subgroups isomorphic to $SO(3)$.
\end{remark}

\iffalse
$\q_1=(x,y,0,0)$ and $\q_2=(0,0,z,w)$ then $d(\q_1, \q_2) =  \sqrt{x^2+y^2+z^2+w^2}$

An isosceles triangle is $\q_1=(r,\theta_1,0,0)$ and $\q_2=(r,\theta_2, 0,0)$ and $\q_3=(0,0, z,w)$, since $d(\q_1, \q_3)= d(\q_2, \q_3)= \sqrt{r^2+z^2+w^2}.$
\fi

Associated to the cotangent-lifted action of $SO(2)_{xy} \times SO(2)_{zw}$
on $T^*\R^4$ is a \textbf{momentum map} $J: T^*\mathbb{R}^4 \to so(2)^* \times so(2)^*$
defined by
$\left<  J(\q, \p),  \xii \right> = \left< \p,  \xii \cdot \q  \right>$.
A straightforward calculation shows:
\begin{align}
&
J(\q, \p)=
\left[
\begin{array}{cccc}
0 & -( p_y  q_x-p_xq_y) &0 & 0\\
(p_y  q_x-p_xq_y) & 0& 0 & 0 \\
0 & 0& 0& -(p_wq_z-p_zq_w)\\
0& 0& (p_wq_z-p_zq_w)  & 0
\end{array}
\right] 
\nonumber 
\end{align}
or,  using the identification $so(2)^* \times so(2)^* \simeq \mathbb{R} \times \mathbb{R}$, %
\begin{align}
 J(\q, \p) = J(q_x, q_y, q_z, q_w, p_x, p_y, p_z, p_w) = (p_y  q_x-p_xq_y, p_wq_z-p_zq_w)\,.
\label{one_mom_map}
\end{align}
If we pass to the \textbf{double-polar coordinates} $(r_1, \theta_1, r_2, \theta_2)$ defined by
\begin{align}
&x= r_1 \cos \theta_1, \quad y= r_1 \sin \theta_1\\
&z= r_2 \cos \theta_2, \quad w= r_2 \sin \theta_2
\end{align}
then the momentum map reads
\begin{align}
 J(\q, \p) = J(r_1, \theta_1, r_2, \theta_2, p_{r_1}, p_{\theta_1}, p_{r_2}, p_{\theta_2})=(p_{\theta_1}, p_{\theta_2}).
\label{two_mom_map}
\end{align}
Let $\muu = (\mu_1, \mu_2) \in so(2)^*\times so(2)^*$  be  fixed momentum value.  For future reference we note that the isotropy group of $\muu$ in the double planar group is given by 
\begin{align*}
\{R \in SO(2)_{xy}\times SO(2)_{zw}\,|\, \text{Ad}^*_{R^{-1}}\hat \muu=\hat \muu\} 
&= \{\hat \muu \in SO(2)_{xy}\times SO(2)_{zw}\,|\, R \hat \muu R^{-1}=\hat \muu\}  \\
&= SO(2)_{xy}\times SO(2)_{zw}.
\end{align*}
\section{Central force problem}\label{sect: central}

Consider the motion of a unit  mass point in $\mathbb{R}^4$ under the influence of a source field located in the origin.  The Hamiltonian is 
\begin{align}
H(\q, \p) =  \frac{1}{2}\p^2 + V(\|\q\|)
\label{orig_Ham}
\end{align}
where $V: D\subseteq \mathbb{R} \to \mathbb{R}$ 
%where $D$
%:=\mathbb{R}^4 \setminus\{\text{possible collisions}\}$ 
is some smooth   potential   depending only on the distance to the origin,
and $D$ some subset of $\mathbb{R}$.
  In  double-polar coordinates we have%
\begin{align}
&H(r_1, \theta_1, r_2, \theta_2, p_{r_1}, p_{\theta_1}, p_{r_2},  p_{\theta_2}) =  \frac{1}{2}\left(p_{r_1}^2 +\frac{p_{\theta_1}^2}{r_1^2} + p_{r_2}^2 +\frac{p_{\theta_2}^2}{r_2^2} \right)+  V\left(\sqrt{r_1^2+r_2^2}\right)\,.
\label{Ham_1}
\end{align}
where $D_1\times D_2 \subseteq (0, \infty) \times  (0, \infty)$  denotes the domain where $V$ is well-defined. Recall that we consider as symmetry group the subgroup $SO(2)_{xy}\times SO(2)_{zw}$. Since the Hamiltonian is invariant under its action, by Noether theorem, we  
obtain  the conservation of the angular momentum, that is
\[\p(t)=(p_{\theta_1}(t), p_{\theta_2}(t))=const.=\muu=(\mu_1,\mu_2)
\]
along any solution. 
We obtain the reduced 2-degrees of freedom system with the Hamiltonian 
\begin{align}
H_{\text{red}}(r_1, p_{r_1}, r_2, p_{r_2}) =  \frac{1}{2}\left(p_{r_1}^2 + p_{r_2}^2 \right)+\frac{\mu_1^2}{2r_1^2}+\frac{\mu_2^2}{2r_2^2}  +V\left(\sqrt{r_1^2+r_2^2}\right)
\label{Ham_2}
\end{align}
coupled with the reconstruction equations:
\begin{align}
\dot \theta_1= \frac{\mu_1^2}{r_1^2(t)}, \,\,\,\,\, \dot \theta_2= \frac{\mu_2^2}{r_2^2(t)}\,.
\label{reconstruction}
\end{align}
\begin{remark}
For $\muu=(\mu_1, \mu_2)\neq \bf{0},$ from the conservation of angular momentum it follows that the projections of the motion on each of the principal planes 
 obey the area laws, and the ratio is these areas is constant. Specifically, let $\displaystyle{A_{xy}(t)=\frac{1}{2}r_1^2(t) \dot \theta_1}$ and $\displaystyle{A_{zw}(t)=\frac{1}{2}r_2^2(t) \dot \theta_2(t)}$ be the  projected areas. Then on each plane we have equal areas in equal times and 
\begin{equation}
\frac{A_{xy}(t)}{A_{zw}(t)} =\left(\frac{\mu_1}{\mu_2}\right)^2  \,\,\,\,\text{for all}\, \,\,t.
\end{equation}
\end{remark}
\begin{remark}(The harmonic oscillator)
For potentials of the form $\displaystyle{V\left(\sqrt{r_1^2+r_2^2} \right) = k\left(r_1^2+r_2^2 \right)}$, $k \in \mathbb{R},$ the system decouples and it is integrable. 
%For more, see Cushman book...
\end{remark}
Since the Hamiltonian is time-independent, the energy is conserved, and so along any solution
$H_{\text{red}}\left(r_1(t), p_{r_1}(t), r_2(t), p_{r_2}(t) \right) =$constant.
The equations of motion read:
\begin{align}
&\dot r_1=p_{r_1}, \,\,\,\,\dot p_{r_1}=\frac{\mu_1^2}{r_1^3}- V'\left(\sqrt{r_1^2+r_2^2}\right)\frac{r_1}{\sqrt{r_1^2+r_2^2}}\\
&\dot r_2=p_{r_2}, \,\,\,\,\dot p_{r_2}=\frac{\mu_2^2}{r_2^3}- V'\left(\sqrt{r_1^2+r_2^2}\right)\frac{r_2}{\sqrt{r_1^2+r_2^2}}\,.
\end{align}
We define the \textit{effective} (or amended) potential
\[\tilde{V}_{\mu_1, \mu_2} (r_1, r_2):=\frac{\mu_1^2}{2r_1^2}+\frac{\mu_2^2}{2r_2^2}  +V\left(\sqrt{r_1^2+r_2^2}\right)
\]
Given the conservation of energy  and since the kinetic energy $\displaystyle{\frac{1}{2}\left(p_{r_1}^2 + p_{r_2}^2 \right)}$ is positive, for a fixed level of energy $h$ we retrieve the allowed (Hill) regions of motion 
\[\mathcal{R}_h(\mu_1, \mu_2) :=\left\{(r_1, \theta_1, r_2,  \theta_2) |\,\tilde{V}_{\mu_1, \mu_2} (r_1, r_2) \leq h\right\}.
\]
For example, if for some $h_0$ fixed, the set $\{(r_1, r_2)\,|\,\tilde{V}_{\mu_1, \mu_2}(r_1, r_2)\leq h_0\} = [a_1, b_1] \times  [a_2, b_2]$, $a_j, b_j \in \mathbb{R}$, $j=1,2$, then all trajectories are bounded and belong to $([a_1, b_1] \times S^1) \times ([a_2, b_2] \times S^1)$ i.e. the product of  two annular regions. 

\subsection{Some invariant manifolds}\label{subsect: inv}

In cartesian coordinates, we observe the sets 
\begin{align}
&{\cal P}_{xy}:=\{(x, y, z, w, p_x, p_y, p_z, p_w) \,|\, z=w=p_z=p_w=0 \} \\
&{\cal P}_{zw}:=\{(x, y, z, w, p_x, p_y, p_z, p_w) \,|\, x=y=p_x=p_y=0 \}
\end{align}
are invariant under the dynamics induced by \eqref{orig_Ham} and that  on each of these the dynamics is given by the standard  planar central force problem.

\smallskip

Let us write the reduced Hamiltonian \eqref{Ham_2} using polar coordinates $r_1= R \cos {\varphi}, r_2 = R \sin {\varphi}:$
\begin{align}
H_{\text{red}}(R,\varphi,  P_R,  P_{\varphi}) =  \frac{P_R^2}{2}
+\frac{P_{\varphi}^2}{2R^2}
+\frac{\mu_1^2}{2R^2\cos^2 \varphi}+\frac{\mu_2^2}{2R^2\sin^2 \varphi}  + V(R)\,,
\label{Ham_cf_polar}
\end{align}
which holds for all $\varphi \neq 0, \pi/2, \pi, 3\pi/2.$
The equations of motion read
\begin{align}
&\dot R=P_R, \quad \,\,\,\,\,\,\,\,\dot P_R= \frac{P_{\varphi}^2}{R^3} + \left(\frac{\mu_1^2}{\cos^2 \varphi}+ \frac{\mu_2^2}{\sin^2 \varphi}   \right)\frac{1}{R^3} - V'(R),\\
&\dot \varphi =  \frac{P_{\varphi}^2}{R^2},  \quad \,\,\,\, \,\,\,\,\dot P_{\varphi}= \left(\frac{\mu_1^2 \sin \varphi}{\cos^3 \varphi} - \frac{\mu_2^2 \cos \varphi}{\sin^3 \varphi}  \right)\frac{1}{R^2}.
\end{align}
For all  non-zero momenta $\muu=(\mu_1, \mu_2) \neq \bf{0}$,  we find that the sets
\begin{align}
{\mathcal M}_{\pm \varphi}:=
\begin{cases}
      & \left\{(R,\varphi,  P_R,  P_{\varphi})\,\Big|\, \tan \varphi = \pm \sqrt{\frac{\mu_2}{\mu_1}}, \,P_{\varphi}=0 \right\},
       \,\,\,\,\text{if}\,\,\,\mu_1\neq0,\,\\
      \\
      &\left\{(R,\varphi,  P_R,  P_{\varphi})\,\Big|\, \tan \varphi = \pm \sqrt{\frac{\mu_1}{\mu_2}}, \,P_{\varphi}=0 \right\}, \,\,\,\,\text{if}\,\,\,\mu_2\neq0
\end{cases}
\end{align}
are  invariant manifolds on which the dynamics is given by the one-degree of freedom system with the Hamiltonian
\begin{align}
\tilde H (R, P_R) =  \frac{P_R^2}{2}
+\frac{\left(\mu_1 +\mu_2 \right)^2}{2R^2} + V(R)\,.
\label{Ham_polar_1_def}
\end{align}
Given that these are  motions  with constant ratio $r_1(t)/r_2(t)$ of the polar radii we introduce
\begin{definition}
Motions with a  constant ratio of the polar radii, that is $r_1(t)= \lambda r_2(t)$ for some $\lambda >0$ and  all $t$, are called  \textbf{proportional motions}.
\end{definition}
Thus we can state that proportional motions form an invariant manifold.
%
%
%only  the distance from the mass point to the origin varies while not changing the direction; in this respect, these could be thought of as  collinear motions in $\mathbb{R}^2$;  we call the dynamics on $M_{i}$ \textit{pseudo-collinear} motions.
Since the system associated to $\tilde H$ had  one-degree of freedom, it is integrable. 
For future reference we introduce
\begin{definition}
A potential $V: D \subset \mathbb{R} \to \mathbb{R}$ is \textbf{attractive} if $V'(\|\q\|)\geq 0$ for all $\q.$
\end{definition}
We also note that 
for any attractive  potential  $V$, $V(R)\neq -1/R^2$, the  equilibria of \eqref{Ham_polar_1_def} are given by
\[
(R, \varphi, P_R, P_\varphi) = \left(R_0, \pm \sqrt{\frac{\mu_2}{\mu_1} },0,0 \right)
\]
where $R_0$ is the root of $ R^3 V'(R)=(\mu_1+\mu_2)^2.$
(As the reader can easily verify, the case of the Jacobi potential $V(R)= -1/R^2$ is degenerate in the sense that either all $R$ are equilibria or there are no equilibria alt all, and will be discussed elsewhere.) Using the reconstruction equations \eqref{reconstruction},  every equilibrium on  ${\mathcal M}_{\pm \varphi}$ corresponds to  orbits  that either are  quasi-periodic, or densely fill in a torus.

\smallskip

Various choices for $V$ lead to different problems.  Below we recall the Kepler problem and find the necessary and sufficient conditions for collisional motion.

\subsection{The Kepler problem}

In this Subsection we consider the classical  Newtonian potential in $\mathbb{R}^3$, that is $V(\q) = -k/\|\q\|,$ $k>0$. As mentioned in the introduction, while this is not the solution of the Laplace equation in $\mathbb{R}^4$ and so from a physical standpoint this is not the $\mathbb{R}^4$-gravitational potential, the induced dynamics is interesting from a theoretical standpoints. 
%Thus this potential was used for instance in \cite{PJ80, PJ81a, PJ81b, AC98, OV06, Ch13}. 
%
\begin{align}
H(\q ,\p) =  \frac{1}{2}\p^2  -\frac{k}{\|\q\|}\,,\,\,\,\,k>0\,.
\label{Ham_Newton_1}
\end{align}
Following \cite{OV06} in this case the  Laplace-Runge-Lenz vector 
\[
\A:=\left(\p^2-\frac{1}{\sqrt{r_1^2+r_2^2}}\right)\q - (\q \cdot \p)\p
\]
is a conserved quantity. Since  it provides 4 integrals of motion, the dynamics drops to an integrable system. 
The dynamics resembles the Kepler problem in $\mathbb{R}^3$; for instance, for allowed negative energies $h<0$, all orbits are ellipses $\displaystyle{ \Theta  \to  \frac{p}{ 1+\epsilon \cos \Theta}}$ where $\displaystyle{p:=\frac{\mu_1^2+\mu_2^2}{k}},$ $\displaystyle{\epsilon:= \sqrt{1+\frac{2h(\mu_1^2+\mu_2^2) }{k^2}}}$ and $\Theta$ is the angle between   $\q$ and $\A;$ for details, see \cite{OV06}.
\begin{proposition}
In  the Kepler problem in $\mathbb{R}^4$,  the motion is collisional  iff $\muu=(\mu_1, \mu_2)=\bf{0}.$
\end{proposition}
\noindent Proof:  For the classical Newtonian potential, the energy conservation 
reads
\begin{align}
\frac{1}{2}\left(p_{r_1}^2(t) + p_{r_2}^2(t)  \right)+\frac{\mu_1^2}{2r_1^2(t) }+\frac{\mu_2^2}{2r_2^2(t)}  -\frac{k}{\sqrt{r_1^2(t) +r_2^2(t) }}=const.=h.
\label{Ham_Newton_cons}
\end{align}
Let $\mu_1^2+\mu_2^2>0$ and let us assume that $\lim \limits_{t \to t^*} \left(r_1^2(t) + r_2^2(t)\right) = 0$ for some $t^* \leq \infty$. Then at least one of $\mu_1$ or $\mu_2$  is non-zero; without loosing generality, say $\mu_1^2\neq0.$ As $t \to t^*$, the left hand side of  \eqref{Ham_Newton_cons} tends to $\infty$, whereas the right hand side  is the finite energy, which is a contradiction.

If $\mu_1=\mu_2=0$ then the reduced Hamiltonian \eqref{Ham_2} (with Newtonian interaction) becomes
\begin{align}
H_{\text{red}}(r_1, p_{r_1}, r_2, p_{r_2}) =  \frac{1}{2}\left(p_{r_1}^2 + p_{r_2}^2 \right) -\frac{k}{\sqrt{r_1^2+r_2^2}}\,.
\label{Ham_Newton_5}
\end{align}
from where we have  the equations of motion
\begin{align}
&\dot r_i=p_{r_i} \\
&\dot p_{r_i}=- \frac{r_1}{(r_1^2+r_2^2)^{3/2}}, \,\,\,\, i=1,2\,.
\end{align}
It is immediate that $r_i(t)$ eventually becomes decreasing for all $t$ greater than some $t^*.$ $\square$

\begin{remark}
The same result is valid if $V(r)=-1/r^{\alpha}$ with $\alpha<2.$
\end{remark}

\begin{proposition}
In the Kepler problem in $\mathbb{R}^4$, collisional orbits  are proportional motions. 
\end{proposition}

\noindent 
Proof: For $\mu=\bf{0},$  in polar coodinates $r_1=R \cos \varphi, r_2= R \sin \varphi$, the  Hamiltonian \eqref{Ham_Newton_5}  reads  
\begin{align}
H_{\text{red}}(R, \varphi, P_R, P_{\varphi}) =  \frac{P_R^2}{2} + \frac{P_{\varphi}^2}{2R^2}-\frac{k}{R}
\label{Ham_Newton_2}
\end{align}
and so it is identical to the classical planar  Kepler problem in polar coordinates.  For the latter,  it is known  that collision is attained only by motions on a straight line, that is those with for $\varphi(t)= $ constant $ =:\varphi_0$. 
Thus  along a collision path in $\mathbb{R}^4$, since  $\tan \varphi_0 =   r_1(t) /  r_2(t)$, the conclusion follows. $\square$

\begin{remark}
Along a collisional paths the coordinates  of the mass point is given by
\begin{align}
x(t)= a R(t), \,\,y(t)= b R(t), \,\,z(t)= c R(t), \,\,w(t)= d R(t), 
\end{align}
where $R(t)$ solves \eqref{Ham_Newton_2} and the constants $a,b,c,d$ are determined by the initial conditions. All collisional paths are collinear.

\end{remark}

%%%%%%%%%%%%%%%%%%%%%%%%%%%%%%%%
\iffalse
\begin{proposition}(Pseudo-collinear motions)
For all $h<0,$ all propostional motions are quasi-periodic.
\end{proposition}
%
\begin{align}
\tilde H (R, P_R) =  \frac{P_R^2}{2}
+\frac{\left(\mu_1 +\mu_2 \right)^2}{2R^2} -\frac{\gamma}{R}\,.
\label{Ham_cf_polar}
\end{align}
\fi
%%%%%%%%%%%%%%%%%%%%%%%%%%%%%%%%%%%

\section{The  n-body problem}\label{n-body_sect}

\subsection{Generalities}

Consider  $n$ points with masses $m_1, m_2, \ldots, m_n$ in $\mathbb{R}^4$ with  mutual interaction  via  some potential.  Their positions  is given by  
 $\q=(\q_1,\q_2, \ldots \q_n) \in Q \subseteq \mathbb{R}^{4n}$ on which  the symmetry group $SO(2) \times SO(2)$ acts diagonally on the principal planes.  Further, the group acts  on $TQ$ and  $T^*Q$ by tangent and co-tangent lift, respectively. The masses (or weights) of the points induce the \textit{mass metric} on $Q$
\begin{align}
\ll \dot \q\,, \dot \r  \gg := \dot \q^t\,\mathbb{M}\, \dot \r\,,\,\,\,\,\, \dot \q, \dot \r \in T\mathbb{R}^{4n}
 \label{Ham_gen_1}
\end{align}
where $\mathbb{M} =\text{diag}(m_1 \mathbb{I}_4 , m_2 \mathbb{I}_4, \ldots, m_n \mathbb{I}_4)\,\,\,\, i=1,2,\ldots, n$, where $\mathbb{I}_4$ is the $4\times 4$ identity matrix, is the \textit{mass matrix}.
 The dynamics is given by the Lagrangian
\begin{align}
L(\q, \dot \q) &=  \frac{1}{2} \dot \q^t \mathbb{M} \,\dot \q - U(\q), \quad \textrm{where}
\label{Lag_gen_1} \\
U(\q)&:= \sum \limits_{1\leq j<k\leq n} V(\|\q_j-\q_k\|).
\label{Lag_gen_2}
\end{align}
 Given the invariance of $L$ to translations, one may prove that the linear momentum is conserved.
Thus,  since  the centre of mass has a rectilinear and uniform motion, we choose without loosing generality the location of the centre of mass to be  the origin. 
\begin{remark}(Collinear, planar and spatial  motions)
Any line in $\mathbb{R}^4$ is an invariant manifold for the dynamics. More precisely if initially all points are on a line with velocities  tangent to that line, then the points will remain on that line at all times. Similar statements are valid for motions in a plane ($\mathbb{R}^2$) or hyperplane ($\mathbb{R}^3$).

\end{remark}
%
%
%\begin{remark}(Planar motions)
%If the initial conditions are such that  all points are in one of the principal planes with velocities tangent to that plane, then the motion will stay on that principal plane at all times.
%\end{remark}

%\begin{remark}(Three body motion)
% Three linearly independent vectors $\q_1, \q_2, \q_3$ (in $\mathbb{R}^4$) determine a plane.  
%\end{remark}
In Hamiltonian formulation,  the dynamics is given by
\begin{align}
H(\q, \p) &=  \frac{1}{2} \p^t \mathbb{M}^{-1} \p + U(\q),\, \quad \textrm{i.e.}\\
H(\q, \p) &=  \frac{1}{2} \p^t \mathbb{M}^{-1} \p + \sum \limits_{1\leq j<k\leq n} V(\|\q_j-\q_k\|)\,
\label{Ham_gen}
\end{align}
for some smooth $v: D\subseteq \mathbb{R} \to \mathbb{R} $
and where the momenta are denoted $\p=(\p_1, \p_2, \ldots, \p_n)$ with $\p_j=(p_{jx}, p_{jy},p_{jz},p_{jw}),$ $j=1,2,\ldots,n.$ The energy is given by the Hamiltonian $H(\q, \p)$ and it is conserved along any solution. The (angular) momentum map is 
\begin{align}
J(\q, \p):= \left( \sum \limits_{j=1}^n (p_{jy}q_{jx} - p_{jx}q_{jy})\,,\sum \limits_{j=1}^n (p_{jw}q_{jz} - p_{jz}q_{jw}) \right)\,.
\end{align}
and by Noether theorem, since $H$ is invariant with respect to the $SO(2)_{xy} \times SO(2)_{zw}$ action, $J$ is conserved as well along any solution.
In double-polar coordinates we have
\begin{align}\label{Hdoublepolar}
H&(\q, \p) =  \frac{1}{2} \sum \limits_{j=1}^n\ \left(p_{r_{j1}}^2  + \frac{p_{\theta_1}^2}{m_jr_{j1}^2} +  p_{r_{j2}}^2  + \frac{p_{\theta_2}^2}{m_jr_{j1}^2} \right) \\
&\quad  \,\,\,\,+ \sum \limits_{1\leq j<k\leq n} 
V\left(\sqrt{r_{j1}^2 + r_{k1}^2 - 2 r_{j1}r_{k1} \cos (\theta_{j1}- \theta_{k1}) +   r_{j2}^2 + r_{k2}^2 - 2 r_{j2}r_{k2} \cos (\theta_{j2}- \theta_{k2} }\, \right).
\nonumber
\end{align}
and 
\begin{align}
J(\q, \p):= \left(  \sum \limits_{j=1}^n p_{\theta_{j1}}\,, \sum \limits_{j=1}^n p_{\theta_{j2}} \right)\,.
\end{align}
\subsection{Relative equilibria are the Albouy-Chenciner balanced configurations}
\label{subsect:RE}

In the following two  subsections we are considering the full symmetry group $SO(4)$.
% as its symmetry group $n$-body problem in $\mathbb{R}^4$  with $SO(4)$ as its symmetry group.
% In this section we are considering the $n$-body problem in $\mathbb{R}^4$  with $SO(4)$ as its symmetry group.

% in Lagrangian formulation given by \eqref{Lag_gen_1} and

\begin{definition}
A solution $(\q(t), \dot \q(t))$  of the $n$-body problem in $\mathbb{R}^4$ as given by the Lagrangian \eqref{Lag_gen_1}  is a \emph{relative equilibrium solution}  if there is \textit{group velocity} $\hat\om \in so(4)$ so that 
\begin{align}\label{E:REdef}
\q(t) =  \text{exp}(t\hat\om)\q_0,\,\,\, \dot \q(t)= \text{exp}(t\hat\om)\dot \q_0
\end{align}
for some \textrm{base point} $\q_0 \in Q$.
Alternatively, given $\hat\om$, an element $\q_0\in Q$
such that \eqref{E:REdef} is a solution
is called 
a \textbf{relative equilibrium (RE) (with group velocity $\hat\om$)}. 
If instead we use the Hamiltonian formulation \eqref{Ham_gen}, a solution $(\q(t), \p(t))$  is a relative equilibrium solution if there is a $\hat\om \in so(4)$  so that 
\begin{align}
\q(t) =  \text{exp}(t\hat\om)\q_0,\,\,\,\p(t)=\text{FL} \left(\text{exp}(t\hat\om)\dot \q_0\right)
\end{align}
for some and $\q_0\in Q$, where $\text{FL}: TQ \to T^*Q$ is the Legendre transform.
\end{definition}

%Note that if $\q(t) =  \text{exp}(t\hat\om)\q(0)$, then $\dot \q(t) =  \text{exp}(t\hat\om)\dot \q(0)$,

\iffalse
\begin{definition}
A solution $(\q(t), \dot \q(t))$  of the $n$-body problem in $\mathbb{R}^4$ as given by the Lagrangian \eqref{Lag_gen_1}  is a \textbf{relative equilibrium solution}  if there is \textit{group velocity} $\hat\om \in so(4)$ and $(\q_0, \dot \q_0)\in TQ$ so that 
%
\begin{align}
\q(t) =  \text{exp}(t\hat\om)\q_0,\,\,\, \dot \q(t)= \text{exp}(t\hat\om)\dot \q_0
\end{align}
%
The configuration $\q_0$ is called to be a \textbf{base point}, whereas the initial condition $(\q_0, \dot \q_0)$ is called a \textbf{relative equilibrium (RE)}. 
If instead we use the Hamiltonian formulation \eqref{Ham_gen}, a solution $(\q(t), \p(t))$  is a relative equilibrium solution if there is a $\hat\om \in so(4)$ and $(\q_0,  \p_0) \in T^*Q$ so that 
%
\begin{align}
\q(t) =  \text{exp}(t\hat\om)\q_0,\,\,\,\p(t)= \text{exp}(t\hat\om)\p_0.
\end{align}
%
whereas $(\q_0,  \p_0)$ is a RE.

\end{definition}
%
\fi

\begin{remark}\label{REnormalform}
Any relative equilibrium with respect to the $SO(4)$ action
is conjugate to a relative equilibrium with respect to the 
action of the double planar rotation group.
Indeed, by Proposition \ref{principalom}, 
for any relative equilibrium $\q_0$ with group velocity $\hat\om \in so(4)$,
there exists an orthogonal matrix $P$ such that 
$P \hat\om P^t \in SO(2)_{xy} \times SO(2)_{zw}$. 
Then the trajectory $\r(t) := P \exp(t\hat\om) \q_0$
satisfies  $(\r(t), \dot \r(t)) = (P \q(t), P \dot \q(t))$,
so it is also a solution of the $n$-body problem, 
and
\[
\r(t) =  P \exp(t\hat\om) P^t \r(0),
= \exp\left(t \, P \hat \om P^t\right) \r_0
= \exp\left(t \, \widehat{P \om} \right) \r_0,
\]
where $\r_0 := \r(0)$.
So $\r_0$ is a relative equlibrium with group velocity $\widehat{P \om}$.
\end{remark}

The following proposition is standard, see e.g. \cite{Ma92}.

\begin{proposition}  There is a RE with base point $\q_0$ and group velocity $\hat\om \in so(4)$ if and only if
$\q_0$ is a critical point of the augmented potential
\begin{align*}
U_{\hat\om}(\q) &:=  U(\q) - \frac{1}{2} \left(\hat\om \cdot \q \right)^t \mathbb{M} (\hat\om \q),
\end{align*}
%where $\om \cdot \q=\hat \om \q$ is the infinitesimal action of $\om$ on $\q$ and 
$\mathbb{M} =\text{diag}(m_1 \mathbb{I}_4 , m_2 \mathbb{I}_4, \ldots, m_n \mathbb{I}_4)$.
\end{proposition}

Note that in \cite{Ma92}, the definition of the augmented potential is more general, 
involving the so-called
\textit{locked inertia tensor}. However for the $n$-body problem, the definition 
reduces to the form shown in the above proposition.

%
%

%
%%%%%%%%%%%%%%%%
\iffalse
By definition,  the \textit{locked inertia tensor}  is the map $\mathbb{I}: \mathbb{R}^{4n} \to \mathit{Lin}(so(2) \times so(2), so(2)^* \times so(2)^*)$  so that $\left<\mathbb{I}(\q)\om, \boldsymbol{\eta} \right>_{so(2)\times so(2)} = \ll \om\cdot \q, \boldsymbol{\eta}\cdot \q \gg$. Using \eqref{inf-gen} we obtain
 %
  %
\begin{align}
\left<\mathbb{I}(\q)\om, \boldsymbol{\eta} \right>_{so(2)\times so(2)} = 
 \sum \limits_{j=1}^n m_j \left( \omega_1 \eta _1 (q_{jx} r_{jx} + q_{jy} r_{jy}  )  + \omega_2 \eta _2 (q_{jz} r_{jz} + q_{jw} r_{jw} )    \right) 
  \end{align}
%
In our case, we are considering the diagonal action of $SO(4)$ on $\R^{4n}$,
so the augmented potential takes the form 
\begin{equation}
V_{\om}(\q) =  V(\q) - \frac{1}{2}  \left(\hat \om \q \right)^t M (\hat \om \q)\label{E:augpot}.
\end{equation}
\fi
%%%%%%%%%%%%%%%%%%%%%%%%%
%
%
%
In our case, we are considering the diagonal action of $SO(4)$ on $\R^{4n}$,
so $\om \cdot \q = \left(\hat \om \q_1, \dots, \hat \om \q_n\right)$, and hence
the augmented potential takes the form 
\begin{align*}
U_{\om}(\q) &=  U(\q) - \frac{1}{2}  \sum_j \left(\hat \om \q_j \right)^t \left(m_j \mathbb{I}_4\right) (\hat \om \q)\\
&= U(\q) - \frac{1}{2} \sum_j m_j \q_j^t\hat \om^t \hat \om \q_j.
\end{align*}
Thus, for all $j=1\dots n$, 
\begin{align*}
\nabla_j U_{\om} (\q)  &= \nabla_j U(\q) -  m_j \hat \om^t \hat \om \q_j \\
&= \nabla_j U(\q) + \hat \om^2 \left( m_j \q_j \right).
\end{align*}
%
%the base points $\q_0$ are determined as solutions of
%
%\begin{align}
%\nabla V(\q) =  - \hat \om^2 \mathbb{M} \q,
%\end{align}
%or equivalently
It follows that $\q$ is a relative equilibrium with group velocity $\hat \om$ if and only if
\begin{align}
\nabla_j U(\q) =  - m_j \hat \om^2 \q_j \quad \textrm{for all } j=1\dots n .\label{E:releq2}
\end{align}
Note that this criterion, which determines a relative equilibrium in the present context,
is equivalent to the following definition introduced by Albouy and Chenciner in \cite{AC98} (see also  \cite{Mo14}):
 \begin{definition} 
 A configuration is \textbf{$4$-balanced}
in $\mathbb{R}^4$ if there is a $4 \times 4$ antisymmetric matrix $A$ such that, for all $j=1\dots n$,
\[
\nabla_j U(\q) =  -m_j A^2  \q_j\, .
\]
\end{definition}
Specifically, a configuration is $4$-balanced if and only if is the base point of a relative equilibrium solution with group velocity given by $A$ in the above definition, i.e. $\hat \om = A$. 
%This connection was noted by Moeckel in his notes \cite{Mo14}.
%

\medskip
Another definition introduced in \cite{AC98} and reproduced in \cite{Mo14} is:
 
\begin{definition}
 A \textbf{central configuration} is an arrangement of the n point masses whose configuration vector satisfies  $\nabla U(\q) +\lambda \mathbb{M} \q=0$
for some real constant $\lambda.$
\end{definition}

This is clearly a special case of the relative equilibrium condition \eqref{E:releq2}, 
with $\hat \om^2 = \lambda \mathbb{I}_4$.
%where $I$ denotes the $4\times 4$ identity matrix.
In fact, as we note below, for $\om = (\omega_1, \omega_2) \in SO(2)_{xy} \times SO(2)_{zw}$,
we have $\hat \om^2 = \mathrm{diag}\left( -\omega_1^2, -\omega_1^2, -\omega_2^2, -\omega_2^2\right)$. Thus
$\q_0$ is a central configuration  if and only if $\q_0$ is a relative equilibrium with 
group velocity 
$\om=(\omega, \omega)$ with $-\omega^2 = \lambda$.

\subsection{Collinear relative equilibria}

An n-point \textit{collinear} configuration, for some $n>1$, is one satisfying
\begin{align} 
\q_j=\lambda_j\q_0, \label{E:collinear}
\end{align}
for some $\lambda_j\neq 0$ for every $j=1,2,\ldots, n$, all distinct, and  $\q_0\in \R^4, \q_0 \ne 0$.
%Without loss of generality, we further assume $\|\q_0\| = 1$.

Recall the relative equilibrium condition \eqref{E:releq2}:
\begin{align*}
\nabla_j U(\q) =  - m_j \hat \om^2 \q_j \quad \textrm{for all } j=1\dots n .
\end{align*}
Given that $U$ has the form $\displaystyle{U(\q) = \sum\limits_{1\leq i<j \leq n} V(\|\q_i - \q_j\|)}$, for a some function $V$, it follows that for every $j$,
\[
\nabla_j U(\q) = \sum_{i\neq j} V'(\|\q_i - \q_j\|) \frac{\q_j - \q_i}{\|\q_i - \q_j\|}.
\]
If $\q$ is collinear, i.e. satisfies \eqref{E:collinear}, then 
\begin{align*}
\nabla_j U(\q) &= \sum_{i\neq j} V'(|\lambda_i - \lambda_j| \|\q_0\|) \frac{(\lambda_j - \lambda_i)\q_0 }{|\lambda_i - \lambda_j| \|\q_0\|} \\
&= \sum_{i\neq j} V'(|\lambda_i - \lambda_j| \|\q_0\|) \,
\mathrm{sign}(\lambda_j - \lambda_i) \frac{\q_0}{\|\q_0\|},
\end{align*}
%
%
%
% NOT: since we are assuming $\|\q_0\| = 1$.
So $\q$ is a relative equilibrium if and only if, for every $j$,
\begin{align}\label{E:REcoll}
\sum_i V'(|\lambda_i - \lambda_j| \|\q_0\|) \,\mathrm{sign}(\lambda_j - \lambda_i)  
\frac{\q_0}{\|\q_0\|}
= - m_j \lambda_j \hat \om^2 \q_0 .
\end{align}
In particular, any collinear relative equilibrium must satisfy: 
$\hat\om^2 \q_0$ is a scalar multiple of $\q_0$.
If $\om = (\omega_1, \omega_2) \in so(2) \times so(2)$
and $\q_0 = \left(\q_0^{xy}, \q_0^{zw}\right)$,
then we have
\begin{align}\label{E:om2}
\hat \om^2 = \left[
\begin{array}{cccc}
0 & -\omega_1 & 0 & 0\\
\omega_1 & 0&0&0\\
0 & 0 & 0& - \omega_2\\
0 & 0 & \omega_2 & 0
\end{array}
\right]^2 =
\left[
\begin{array}{cccc}
-\omega_1^2 & 0 & 0 & 0\\
0 & -\omega_1^2 &0&0\\
0 & 0 & - \omega_2^2 & 0\\
0 & 0 & 0 & - \omega_2^2
\end{array}
\right] 
\end{align}
and
\begin{align}\label{E:om2q0}
\hat\om^2 \q_0 
= \left[\begin{array}{c} -\omega_1^2 \q_0^{xy} \\ -\omega_2^2 \q_0^{zw} \end{array} \right],
\end{align}
which leads us to the following proposition.

\begin{proposition}\label{prop:coll}
Let $\q$ be a collinear configuration, written as above, i.e. 
\[
\q = \left(\q_1,\dots,\q_n\right) = \left(\lambda_1 \q_0,\dots,\lambda_n \q_0\right), %\textrm{ with } q_j = \lambda_j \q_0
\]
%(so in particular $\q$ is nonzero), 
for some distinct nonzero $\lambda_j$'s and some nonzero $\q_0 \in \R^4$,
and let
$\om = \left(\omega_1, \omega_2\right) \in so(2) \times so(2)$.
Then $\q$ is the base point of a RE with group velocity $\hat \om$ if and only if 
%In both cases, the RE condition \eqref{E:REcoll} is equivalent to 
the following scalar condition is satisfied 
\begin{align}\label{E:REcollscalar}
\sum_i V'(|\lambda_i - \lambda_j| \|\q_0\|) \, \mathrm{sign}(\lambda_j - \lambda_i)
= \omega^2 m_j \lambda_j \|\q_0\|
\end{align}
and
one
of the following holds:
\begin{enumerate}
\item $\q_0$ lies in the principal plane $Oxy$, and $\omega = \omega_1$;
\item $\q_0$ lies in the principal plane $Ozw$, and $\omega = \omega_2$;
\item $\q_0$ lies in neither of the principal planes, and $\omega = \omega_1 = \omega_2$.
\end{enumerate}
In all cases, there exists a $2$-dimensional subspace in which the RE remains, i.e.
$\q_j(t)$ remains in the subspace for all $j$ and all $t$. 
In cases (i) and (ii) this is the principal plane containing $\q_0$.
In all cases, the RE has underdetermined group velocity due to isotropy.
In case (i) the group velocity is $\widehat{\left(\omega_1, *\right)}$ for $*$ an arbitrary element of $so(2)$,
while in case (ii) it is
$\widehat{\left(*, \omega_2\right)}$.
%In case (iii), there exists an orthogonal change of coordinates  mapping the subspace onto
%a principal plane. 

In case (iii), 
there are scaled projections of the RE into each of the principal planes that are themselves REs.
Specifically, let $\q_0 = \left(\q_0^{xy}, \q_0^{zw}\right)$, and similarly for each $\q_j$,
%and let $\q^{xy} = \left(\q_1^{xy}, \dots \q_n^{xy}\right)$.
and define 
\[
\mathrm{proj}_{xy} \q = \left(\left(\q_1^{xy}, 0\right), \dots \left(\q_n^{xy}, 0\right)\right)
= \left(\lambda_1 \left(\q_0^{xy}, 0\right), \dots \lambda_n  \left(\q_0^{xy}, 0\right)\right),
\]
and similarly $\mathrm{proj}_{zw} \q$.
%an equivalent condition is that 
%there exists a \textit{planar} relative equilibrium $\r = \left(\lambda_1 \r_0, \dots, \lambda_n 
%\r_0\right)$
%with angular velocity $\omega$, such that $\| \r_0 \| = \| \q_0\|$.
%Note that in this case, $\q^{xy}$ and $\q^{zw}$ are not relative equilibria with angular velocity $\omega$, however their scalar multiples
Then both $\displaystyle{\frac{\|\q_0\|}{\|\q_0^{xy}\|} \mathrm{proj}_{xy} \q }$ and
$\displaystyle{\frac{\|\q_0\|}{\|\q_0^{zw}\|} \mathrm{proj}_{zw}\q }$ are relative equilibria with group velocity 
$\widehat{(\omega, *)}$ or $\widehat{(*, \omega)}$, respectively.
%where we denote by  $\q_0^{xy}$ and $\q_0^{xy}$ the projections of of $\q_0$ on the planes $Oxy$ and $Ozw,$ respectively.
%(Note that $\q_0^{xy}$ and $\q_0^{zw}$ need not be parallel.)

In general, for $\om \in so(4)$, there exists an orthogonal change of coordinates that brings $\omega$ into
$so(2) \times so(2)$, so the above analysis applies. In particular, 
for any collinear RE there exists a $2$-dimensional subspace in which the RE remains.
%FALSE Any two RE having the same magnitude of the group velocity are equivalent by an orthogonal transformation.
\end{proposition}
%%%%%%%%%%%%%%%%%%%%%%%%%%%
%
%
\noindent
Proof: 
As observed above, the RE condition \eqref{E:REcoll} implies that 
$\hat\om^2 \q_0 = k \q_0$ for some $k \in \R$, and from
\eqref{E:om2q0}, this occurs
if and only if 
one of the following is true:
(i) $\q_0^{xy} = 0$ and $k = -\omega_2^2$;
(ii) $\q_0^{zw} = 0$ and $k = -\omega_1^2$;
(iii) $\q_0^{xy} \ne 0, \q_0^{zw} \ne 0$ and $\omega_1 = \omega_2 =: \omega$ and $k= -\omega^2$.
%In the first two cases, $\q_0$ lies in one of the principal planes.
In all three cases, 
\eqref{E:REcoll} reduces to the scalar condition \eqref{E:REcollscalar},
with $\omega$ defined as in the proposition.

If $\q_0$ is in a principal plane, then each $\q_j$ is also in that plane, and since
the principal planes are invariant under the $SO(2)\times SO(2)$ action, any RE remains in that plane.
Note that if $\q_0 \in Oxy$ then it has isotropy group $SO(2)_{zw}$, so
the second component of the angular velocity $\omega_2$ is undetermined; similarly if $\q_0 \in Ozw$
then $\omega_1$ is undetermined.
%If $\q_0$ is in neither of the principal planes, then we have already seen that $\omega_1= \omega_2 =: \omega$,
%so $\hat \om$ has eigenvalues $\pm \omega$, each of multiplicity $2$, and in particular, no zero eigenvalues. 
%If there were some 3-dimensional subspace $V$ in which corresponding RE remained, then $\hat \om$
%would have a zero eigenvalue with eigenvector in the direction orthogonal to $V$, giving a contradiction.

In case (iii), the RE still remains in a fixed $2$-dimensional subspace. 
This follows from the special form of the group velocity $\om = \left(\omega, \omega\right)$,
which says that both projections of $\q$ onto the principal planes rotate at the same speed.
Indeed, let
$\displaystyle \v_0^{xy} = \left[\begin{array}{cc} 0 & -1 \\ 1 & 0\end{array}\right] \q^{xy}_0$
and $\displaystyle \v_0^{zw} = \left[\begin{array}{cc} 0 & -1 \\ 1 & 0\end{array}\right] \q^{zw}_0$
and $\v_0 = \left(\v_0^{xy}, \v_0^{zw}\right).$
Then for any $t$ and $j$,
\begin{align*}
\exp \left(t \hat \om\right) \q_j &= \lambda_j \exp \left(t \hat \om\right) \q_0
= \lambda_j \left(\exp (t \omega)  \q_0^{xy}, \exp (t \omega) \q_0^{zw}\right) \\
&= \lambda_j 
\left( \cos (t \omega) \q_0^{xy} + \sin (t \omega) \v_0^{xy}, \cos (t \omega) \q_0^{zw} + \sin (t \omega) \v_0^{zw}\right) \\
& = \lambda_j \left(\cos (t \omega) \q_0 + \sin (t \omega) \v_0\right).
\end{align*}
Hence
\begin{align*}
\exp \left(t \hat \om\right) \cdot \q &= 
\left(\lambda_1 \left(\cos (t \omega) \q_0 + \sin (t \omega) \v_0\right), \dots, \lambda_n \left(\cos (t \omega) \q_0 + \sin (t \omega) \v_0\right)\right) \\
& = \cos (t \omega) \q + \sin(t \omega) \v,
\quad \textrm{ where } \v := \left(\lambda_1 \v_0, \dots, \lambda_n \v_0\right).
\end{align*}
Thus the RE remains in the $2$-dimensional subspace spanned by $\q$ and $\v$.

Finally, in case (iii), let $\displaystyle{\r_0 = \frac{\|\q_0\|}{\|\q_0^{xy}\|} \left(\q_0^{xy}, 0\right) }$, and note that
$\|\r_0\| = \|\q_0\|$ and 
$\frac{\|\q_0\|}{\|\q_0^{xy}\|} \mathrm{proj}_{xy} \q = \left( \lambda_1 \r_0, \dots,  \lambda_n \r_0 \right)$.
Therefore the scalar condition \eqref{E:REcollscalar} for $\q$ and $\hat \om$ is
\textit{the same} as one that arises from 
$\frac{\|\q_0\|}{\|\q_0^{xy}\|} \mathrm{proj}_{xy} \q$ and $\widehat{(\omega, *)}$ (where the asterisk
indicates an arbitrary value). A similar argument applies to  the projections into the $zw$ plane.
Hence both $\displaystyle{\frac{\|\q_0\|}{\|\q_0^{xy}\|} \mathrm{proj}_{xy} \q }$ and
$\displaystyle{\frac{\|\q_0\|}{\|\q_0^{zw}\|} \mathrm{proj}_{zw}\q }$ are relative equilibria with group velocity 
$\widehat{(\omega, *)}$ or $\widehat{(*, \omega)}$, respectively.
$\square$

\begin{remark} 
%Consider an $n$-body problem in $\R^4$ with an $2$-dimensional subspace that is invariant under the dynamics.
%Suppose the
%For potentials that depend only on distances between bodies, the restriction of the $n$-body problem in $\R^4$
%to any invariant $2$-dimensional subspace is an $n$-body problem with a potential of the same form.
Since all collinear REs in $\R^4$ remain in a $2$-dimensional subspace,
they have the same configurations as in the corresponding $n$-body problem in $\R^2$.

In $\R^3$, each RE of the $n$-body problem remains in a plane perpendicular to the angular velocity vector.
The situation is almost the same for the $n$-body problem in $\R^4$, 
where  ``angular velocity'' is replaced by ``group velocity''.
Each RE still moves in a fixed plane, however due to isotropy, it is not always possible to determine this plane
from the group velocity $\hat \om$ alone. (Consider the case (iii) above.)
\end{remark}

\iffalse
\begin{remark}
The Newtonian $n$-body problem in $\R^4$ has the same collinear REs as in the classical 
Newtonian $3$-body problem: one RE for each of the $3$ possible orderings of the bodies
along a line.
Any two RE having the same magnitude of angular velocity are 
equivalent by a rotation.
\end{remark}

\begin{remark}
In case the potential admits a collinear \emph{equilibrium}, corresponding to $\omega=0$,
(which is not the case for the Newtonian $n$-body problem, but is true 
for molecular potentials), then the families of relative equilibria described above
collapse to a single equilibrium point, which is thus a bifurcation point.
\end{remark}
\fi

\section{Regular n-gons and  Discrete Reduction}\label{sect:DR}

When the masses are equal, due to the finite symmetries, we are able to detect  low dimensional invariant manifolds using the method of \textit{Discrete Reduction} (reviewed below). These invariant manifolds will consist of equilibria and relative equilibria 
with configurations that are regular $n$-gons, so we begin with a general discussion of these in $\R^4$.
We consider an $n$-gon to be an n-tuple of points, which are the vertices.
By ``regular'' we do not simply mean that the side lengths are equal (which would include e.g. all rhombuses), nor do we
wish to consider only planar shapes. Instead, we follow Coxeter \cite{Cox36}: 
``A polygon (which may be skew) is said to be regular if it
possesses a symmetry which cyclically permutes the vertices (and therefore
also the sides) of the polygon.'' Here, ``skew'' means nonplanar. The symmetry transformation is required to belong to a predefined group;
for example $SO(2)$ gives rise to the usual planar regular polygons, centred at the origin.
The requirement of a \textit{cyclic} permutation means that the symmetry transformation, which we call the \textit{symmetry generator}, must have order equal to the number of points n. (Recall that the \textit{order} of a group element $g$ is the smallest $k$ such that $g^k=e$, the identity.)
%Given a general symmetry group $G$, a polygon is regular if there exists a $g\in G$that cyclically permutes its vertices.
Since we are defining polygons by n-tuples $\left(\q_1,\dots,\q_n\right)$, the order in which the points are listed matters;
but by relabelling them if necessary, we may assume that there exists a symmetry transformation $g$ such that
$g\q_i = \q_{i+1(\mathrm{mod} n)}$ for all $i$.
In summary, we arrive at:

\begin{definition}\label{def:regular}
A \textit{regular $n$-gon} in $\R^d$, with respect to a group of transformations $G$, is a configuration $\q = \left(\q_1,\dots,\q_n\right)$, 
with distinct points $\q_i$, such that there exists 
an element $g\in G$
of order n, called the symmetry generator, such that $g\q_i = \q_{i+1(\mathrm{mod} n)}$ for all $i$.
% or equivalently,
%$\q = \left(\q_1,g \q_1, g^2 \q_1, \dots, g^{n-1}\q_1\right)$
\end{definition}

In our application $G=SO(4)$, and we have seen that every element of $SO(4)$ is conjugate to an element of 
the double planar rotation group $SO(2)_{xy}\times SO(2)_{zw}$. 
Thus it suffices to consider $G=SO(2)_{xy}\times SO(2)_{zw}$. The finite order elements of this group are those of the form
$R = \left(R_\frac{2\pi a_1}{b_1}, R_\frac{2\pi a_2}{b_2}\right)$, for positive integers $a_1,b_1,a_2,b_2$,
where $R_\theta$ is counterclockwise rotation by angle $\theta$.
Without loss of generality we assume 
that $a_j$ and $b_j$ are relatively prime, for $j=1,2$, so that the order of $R_\frac{2\pi a_j}{b_j}$ is $b_j$.
It follows that the order of $R$ is $lcm(b_1,b_2)$, the least common multiple of $b_1$ and $b_2$.
Note that if $b_1=b_2=n$, then the projections of the polygon onto each principal plane are also regular $n$-gons.
Otherwise the projections are not injective, so the distinctness requirement in the above definition is not satisfied. 
However in general, $\left(\q_{11}, \dots, \q_{bj}\right)$ is a regular $b_j$-gon.
The projection of the original $n$-gon onto the $xy$-plane covers the $b_j$-gon $n/b_j$ times.

Since $a_j$ and $b_j$ are relatively prime, the group generated by $R_\frac{2\pi a_j}{b_j}$ is also generated by $R_\frac{2\pi}{b_j}$.
A $b_j$-gon in $\R^2$ with symmetry generator $R_\frac{2\pi}{b_j}$ is \textit{convex}, meaning that the vertices are numbered in order around the polygon,
and so joining the points in order gives a curve that bounds a convex set.
A $b_j$-gon with symmetry generator $R_\frac{2\pi a_j}{b_j}$, for $a_j\in {2,\dots,b_j-2}$ is non-convex, i.e. ``star-shaped'', e.g. a pentagram.
In this case it is always possible to relabel the points of the $b_j$-gon so that it has symmetry generator   $R_\frac{2\pi}{b_j}$ and is convex.
Given an $n$-gon in $\R^4$ that projects to $n$-gons in both principal planes,
by relabelling points if needed, it is possible to make either one of the two $n$-gons convex, but not necessarily both at once.
We call an $n$-gon \textit{synchronised} if it is possible to make both projections convex at once.

\begin{proposition}\label{prop:reg} Let $n\ge 2$.
There are three types of regular $n$-gon in $\R^4$ with respect to $SO(2)_{xy}\times SO(2)_{zw}$ (and hence w.r.t. 
$SO(4)$ as well):
\begin{enumerate}
\item Planar, lying in a 2-dimensional plane in $\R^4$. Either the $n$-gon lies in a principal plane, or it projects to similar $n$-gons in each of the principal planes.
These two $n$-gons are \textit{synchronised}, meaning that there exists a labelling of the points such that each of the projected $n$-gons is convex.
\item Nonplanar, i.e. skew. The projections of the $n$-gon onto the two principal planes are not similar.
\begin{enumerate}
\item Type I (\textit{unsynchronised} $n$-gons): Both projections are $n$-gons, but there does not exist
a labelling of the points that makes both projections convex simultaneously. 
\item Type II (lower order projections): The projection of the $n$-gon onto at least one principal plane is a regular 
$b_j$-gon for some $b_j<n$.
In this case, %the projection onto the other plane is a regular $q$-gon, for some $q\le n$ with 
$n = lcm(b_1,b_2)$.
\end{enumerate}
\end{enumerate}
\end{proposition}

\begin{proof} Let $\q = \left(\q_1,\dots,\q_n\right)$ be a regular $n$-gon in $\R^4$.
If the $n$-gon is planar, then for every $i$, $\q_{i}$ is a linear combination of $\q_{1}$ and $\q_{2}$, 
which implies that 
$\q_{ij}$ is the \textit{same} linear combination of $\q_{1j}$ and $\q_{2j}$ for $j=1,2$, so the two projected polygons are similar to each other,
meaning that there exists an invertible linear transformation taking one to the other;
since both projections are regular polygons, this just means that the order of the points on the polygon is the same in both projections.
If the two projections of the $n$-gon are similar $b$-gons, with $b_1=b_2=b$, then $b = n$, since we must have 
$n = lcm(b, b)$.
There exists a (re)labelling of the original points $\left(\q_1,\dots,\q_n\right)$ such that each of the projected $n$-gons is convex,
since this is always true of a single $n$-gon, and the two $n$-gons are congruent.
Conversely, suppose that both projections of the $n$-gons are also $n$-gons, i.e. for each $j=1,2$ the projected points $\left(\q_{1j},\dots,\q_{nj}\right)$ are distinct.
Suppose there exists a (re)labelling of the original points $\left(\q_1,\dots,\q_n\right)$ such that both of the projected $n$-gons are convex.
Then they are congruent.
%Then $\left(R_\frac{2\pi}{n}, R_\frac{2\pi}{n}\right)$ is a symmetry generator. %, and the two projected $n$-gons are congruent.
For each $j=1,2$, each projected point $\q_{ij}$ is a linear combination of $\q_{1j}$ and $\q_{2j}$.
Since the two projected $n$-gons are congruent, the coefficients in this combination are 
independent of $j$, so each $\q_{i}$ is a linear combination of $\q_{1}$ and $\q_{2}$.
This shows that the $n$-gon is planar.

The above argument also shows that if the two projections are unsynchronised $n$-gons, then the original $n$-gon $\q$ is nonplanar.
Finally, we consider the case (Type II) in which at least one of the projections of $\q$ is a $b$-gon for some $b<n$.
In this case it is not possible for the two projections to be congruent.
\end{proof}

\begin{remark}
Regular planar polygons may have any number of vertices.
The smallest nonplanar polygon of type I has $n=5$ vertices, with one projection convex and the other a pentagram.
The smallest nonplanar polygon of type II has $n=4$ vertices, with one projection a square and the other a digon (i.e. having $2$ vertices). 
The next smallest nonplanar polygon of type II has $n=6$ vertices, with one projection a triangle and the other a digon. 
The three examples are the only nonplanar polygons with fewer than $7$ vertices.
\end{remark}

\medskip

We now briefly recall the theory of Discrete Reduction as presented in \cite{Ma92} (but also see \cite{GS84} pp.203). 
Let $\Sigma$ be a finite  group acting on a cotangent bundle $T^*Q$, and consider its fixed point set:
\begin{equation}
\text{Fix}\,(\Sigma, T^*Q):= \{(q,p) \in T^*Q\,|\, g(q,p)=(q,p)\,\,\,\forall\, \sigma \in \Sigma\}\,.
 \end{equation}
If a Hamiltonian $H: T^*Q \to \mathbb{R}$  is $\Sigma$-invariant and the symplectic structure is preserved under the $\Sigma$-action, then $\text{Fix}\,(\Sigma, T^*Q)$ is a symplectic submanifold of $T^*Q$ and it is an invariant  manifold for the dynamics of $H$.
Moreover, if the symplectic structure and $H$ are also invariant under the action of a Lie group $G$ giving rise to an equivariant momentum map, and  
the actions of $G$ and $\Sigma$ commute (or satisfy a more general compatibility condition)
%acts on $G$ in such a way that the actions  of $G$ an $\Sigma$ are compatible 
%{\color{red} (does this just mean that the two actions commute?)}
then $\displaystyle{H\big|_{\text{Fix}\,(\Sigma, T^*Q)} \to \mathbb{R}}$
% is a co-tangent bundle Hamiltonian system with 
is a $G$ symmetric Hamiltonian system.
%Moreover, 
By  \textit{Palais' Principle of Criticality} \cite{Pal79} any equilibrium or RE in $\text{Fix}\,(\Sigma, T^*Q)$ is also an equilibrium or RE, respectively, in the full $T^*Q$ phase space,
where the $RE$ are with respect to the same commuting group $G$ mentioned above, if such a group exists.

\smallskip
Returning to the dynamics of $n$ mass points  in $\mathbb{R}^4$, let us assume all masses to be unity, i.e.  $m_1=m_2=\ldots=m_n=1$.
We consider regular $n$-gon configurations in $\R^4$, centred at the origin. 
As discussed above, these configurations have associated symmetry generators $R\in SO(4)$,
$R = \left(R_\frac{2\pi a_1}{b_1}, R_\frac{2\pi a_2}{b_2}\right)$, for positive integers $a_1,b_1,a_2,b_2$,
with $a_j$ and $b_j$ relatively prime, for $j=1,2$,
and $n = lcm(b_1,b_2)$.
Once the symmetry generator $R$ is specified, the set of all $n$-gon configurations with this symmetry
is:
\begin{align}\label{E:ngonR}
\{ \left(\q_1, \dots, \q_n\right) \, : \, R\q_i = \q_{i+1 (\mathrm{mod} n)} \textrm{ for all } i\}.
\end{align}
We can express this set in terms of a group action as follows.
For every $k \in \{0,\dots, n-1\}$, let $\sigma_k$ be the cyclic permutation of $\{1, \dots, n\}$
given by $\sigma_k(i) = (i-k) (\textrm{mod} \,n)$, i.e. a ``backwards shift'' by $k$ positions. Let 
$C_n = \mathbb{Z}_n = \{0, \dots, n-1\}$ be the finite cyclic group of order n, and 
let $C_n$ act on $ Q = \left(\R^4\right)^n \simeq \left(\R^2\times \R^2\right)^n$
 as follows:
\begin{align*}
k \cdot \left(\q_1, \dots, \q_n\right) 
&= \left({\cal R}^k \q_{\sigma_{k}(1)}, \dots,
{\cal R}^k \q_{\sigma_{k}(n)}\right) .
\end{align*}
By abuse of notation, we denote this action by $C_n^{\cal R}$. 
Note that in particular, $1 \cdot \left(\q_1, \dots, \q_n\right) 
= \left({\cal R} \q_{\sigma_{1}(1)}, \dots,
{\cal R} \q_{\sigma_{1}(n)}\right)$, and this generates the whole action.
So
\begin{align*}
\mathrm{Fix}\,\left(C_n^{\cal R}, Q\right) 
&= 
\{ \left(\q_1, \dots, \q_n\right) \, : \, \left(\q_1, \dots, \q_n\right) = \left({\cal R} \q_{\sigma_{1}(1)}, \dots,
{\cal R} \q_{\sigma_{1}(n)}\right)
\} \\
&= \{ \left(\q_1, \dots, \q_n\right) \, : \, R\q_i = \q_{\sigma_1^{-1}(i)} \textrm{ for all } i\} \\
&= \{ \left(\q_1, \dots, \q_n\right) \, : \, R\q_i = \q_{i+1 (\mathrm{mod} n)} \textrm{ for all } i\},
\end{align*}
which is the same expression as above in \eqref{E:ngonR}.
Thus
$\text{Fix}\,\left(C_n^{\cal R}, Q\right)$ consists of all
regular $n$-gon configurations with symmetry generator $R$.
Let $\hat Q^R=\text{Fix}\,\left(C_n^R, Q\right)$. 
In polar coordinates, the action $C_n^R$ is expressed as:
%On the configuration space $Q \subseteq \mathbb{R}^{4n}$, the product  group $(S_n \times C_n) \times (S_n \times C_n)$  acts diagonally on each copy of $\mathbb{R}^4$. 
%
\begin{align*}
k \cdot &\left(r_{11}, \theta_{11}, r_{12}, \theta_{12}, \dots,  r_{n1}, \theta_{n1}, r_{n2}, \theta_{n2}
\right)
= \left(r_{\sigma_{k}(1)1}, \theta_{\sigma_{k}(1)1}+ \frac{2\pi a_1 k}{b_1},
r_{\sigma_{k}(1)2}, \theta_{\sigma_{k}(1)2}+ \frac{2\pi a_2k}{b_2}, \right. \\
&\dots,
\left. r_{\sigma_{k}(n)1}, \theta_{\sigma_{k}(n)1}+ \frac{2\pi a_1 k}{b_1},
r_{\sigma_{l}(n)2}, \theta_{\sigma_{l}(n)2}+ \frac{2\pi a_2k}{b_2}
\right).
\end{align*}

When $R = \left(R_\frac{2\pi}{n}, R_\frac{2\pi}{n}\right)$, which corresponds to planar $n$-gons,
we have:
\begin{align*}
\textrm{(planar case)} \quad \hat Q^R
=\left\{ \q \in Q\,|\,r_{ij} = r_{kj},  \,\, \theta_{i j}-\theta_{\sigma_1(i)j} = \frac{2\pi}{n}, \,\,\textrm{ for all } i,k =1,2,\ldots n, 
\,\,j=1,2
\right\},
\end{align*}
%From Proposition \ref{prop:reg} there are different classes of regular polygons 
i.e. all configurations such that both projections are regular $n$-gons.
Note that there is no restriction on the relationship between these two $n$-gons:
they may have different radii and/or different phases $\theta_{ij}$.
%Clearly, $\hat Q$ is determined by the projections of $\q_1$ in the two principal planes.

In all cases (planar or nonplanar), every element of $\hat Q^R$ is determined by the two common radii $r_1$ and $r_2$, defined by
$r_j := r_{kj}$ for any $k$; and the phases of the first point, $\theta_1 := \theta_{11}$ 
and $\theta_2 := \theta_{12}$. 
This observation defines double-polar coordinates $(r_1,\theta_1, r_2,\theta_2)$ on $\hat Q^R$.
%\[
%\hat Q \simeq \{ (r_1,\theta_1, r_2,\theta_2) \,|\, r_1, r_2 > 0, \theta_1, \theta_2 \in \R\} 
%\]
The action of $C_n^R$ lifts to an action on $T^*Q$, expressed in polar coordinates as
\begin{align*}
&(k, l) \cdot \left(r_{11}, \theta_{11},  p_{r_{11}}, p_{\theta_{11}}, r_{12},\theta_{12}, p_{r_{12}}, p_{\theta_{12}},
\dots, r_{n1}, \theta_{n1},  p_{r_{n1}}, p_{\theta_{n1}}, r_{n2},\theta_{n2}, p_{r_{n2}}, p_{\theta_{n2}} \right) \\
&= (r_{\sigma_{k}(1)1}, \theta_{\sigma_{k}(1)1}+ \frac{2\pi a_1 k}{b_1},p_{r_{\sigma_k(1)1}}, p_{\theta_{\sigma_k(1)1}},
r_{\sigma_{k}(1)2}, \theta_{\sigma_{k}(1)2}+ \frac{2\pi a_2k}{b_2},p_{r_{\sigma_k(1)2}}, p_{\theta_{\sigma_k(1)2}}, \dots,\\
&r_{\sigma_{k}(n)1}, \theta_{\sigma_{k}(n)1}+ \frac{2\pi a_1k}{b_1},p_{r_{\sigma_k(n)1}}, p_{\theta_{\sigma_k(n)1}},
r_{\sigma_{k}(n)2}, \theta_{\sigma_{k}(n)2}+ \frac{2\pi a_2}{b_2},p_{r_{\sigma_k(n)2}}, p_{\theta_{\sigma_k(n)2}}, ).
\end{align*}
The fixed point set in $T^*Q$ is 
\begin{align*}
%\textrm{Fix}(C^n\times C^n, T^* Q)
%&\left\{ (\q, \p) \in T^*Q \,|\,r_{ij} = r_{kj}, \theta_{ij}- \theta_{\sigma_1(i)j} = \frac{2\pi}{n},  p_{ij}= p_{kj}\,, \,\, \textrm{for all } i,k=1,2,\ldots n, 
%\,\,j=1,2 \right\}
%&=
\left\{ (\q, \p) \in T^*Q \,|\, \q \in \hat Q \textrm{ and }  p_{r_{ij}}= p_{r_{kj}}  
\textrm{ and } p_{\theta_{ij}}= p_{\theta_{kj}}\,, \,\, \textrm{for all } i,k=1,2,\ldots n, 
\,\,j=1,2 \right\},
\end{align*}
i.e. the momenta of all points are required to be equal when expressed in polar coordinates.
Note that this equals $T^*\hat Q$, which can be checked directly, and also follows from a general result \cite{Ma92}.
By  the Discrete Reduction method,  $T^*\hat Q$  is a symplectic invariant manifold and the dynamics on it is given by the restriction $\hat H$   of the Hamiltonian \eqref{Ham_gen} to   $T^*\hat Q$. Also, any RE of $\hat H$ is a RE of the full system.

Any \textit{planar} configuration in $\hat Q$ projects to two regular $n$-gons with side lengths
$2 r_j\,  \sin \frac{\pi}{n}$,
where $r_j$ is the common radius of the $j^{th}$ $n$-gon, for $j=1,2$. So the original $n$-gon has side length
\begin{align}
l =  2\sqrt{r_1^2 + r_2^2}\,  \sin \frac{\pi}{n} ,
\label{side}
\end{align}
The distance between any two points on the $n$-gon, with angle $\frac{k\pi}{n}$
between them, is 
\begin{align}
\textrm{(planar case)} \quad
2\sqrt{r_1^2 + r_2^2} \left|\sin \frac{k\pi}{n}\right|.
\end{align}
The general formula, for $R = \left(R_\frac{2\pi a_1}{b_1}, R_\frac{2\pi a_2}{b_2}\right)$,
for the distance between $\q$ and $R^k \q$, is
\begin{align}
2 \sqrt{r_1^2 \sin^2 \frac{ka_1\pi}{b_1}+ r_2^2 \sin^2 \frac{ka_2\pi}{b_2}}.
\end{align}

As a consequence, 
we have the following formula for the restricted Hamiltonian,
where for $j=1,2$, the momentum of the $j^{th}$ $n$-gon 
in polar coordinates is $\left(p_{r_j}, p_{\theta_j}\right)$:
\begin{align}
&\hat H(r_1, \theta_1, p_{r_1}, p_{\theta_1}, r_2, \theta_2, p_{r_2}, p_{\theta_2}) \notag \\
&= \frac{n}{2}\left( p_{r_1}^2 +  \frac{p_{\theta_1}^2}{r_1^2}  + p_{r_2}^2 +  \frac{p_{\theta_2}^2}{r_2^2} \right)
+   \sum \limits_{1\leq k<l\leq n}
V \left( 
2\sqrt{
r_1^2 \sin^2 \frac{(l-k)a_1\pi}{b_1} +  r_2^2 \sin^2 \frac{(l-k)a_2 \pi}{b_2}
}
 \right)
 \\
 &= \frac{n}{2}\left( p_{r_1}^2 +  \frac{p_{\theta_1}^2}{r_1^2}  + p_{r_2}^2 +  \frac{p_{\theta_2}^2}{r_2^2} \right)
+   \frac{n}{2} \sum \limits_{1\leq k < n}
V \left( 
2\sqrt{
r_1^2 \sin^2 \frac{ka_1\pi}{b_1} +  r_2^2 \sin^2 \frac{ka_2 \pi}{b_2}}
% r_1^2 + r_2^2} \, \left| \sin \frac{k \pi}{n}\right|
 \right).
\label{H_discrete}
\end{align}
In the \textrm{planar} case, this simplifies to:
\begin{align}
&\hat H(r_1, \theta_1, p_{r_1}, p_{\theta_1}, r_2, \theta_2, p_{r_2}, p_{\theta_2}) \notag \\
&= \frac{n}{2}\left( p_{r_1}^2 +  \frac{p_{\theta_1}^2}{r_1^2}  + p_{r_2}^2 +  \frac{p_{\theta_2}^2}{r_2^2} \right)
+   \frac{n}{2} \sum \limits_{1\leq k < n}
V \left( 
2\sqrt{
 r_1^2 + r_2^2}
  \left|\sin \frac{k\pi}{n}\right| \, 
 \right).
\label{H_discreteplanar}
\end{align}

\begin{remark}
If  the angular velocity  is zero, the dynamics on $T^*\hat Q$ consist of homothetic solutions.
\end{remark}
%
%
% Note that  we have n distinct -homographic invariant manifolds, one for each $m$ fixed. 

%
%
It is immediate that due to the $SO(2)_{xy}\times SO(2)_{zw}$-invariance of $\hat H,$ the angular momentum is conserved and so $\left(p_{\theta_1}(t), p_{\theta_2}(t)\right)=(c_1, c_2)$ for all $t.$  Note that $n c_j=\mu_j,$ $j=1,2,$ where $\muu=(\mu_1,\mu_2)$ is the total angular momentum of the (equal mass) $n$-body problem. The dynamics on $\hat H$ is reducible to a two degrees of freedom system and the reduced Hamiltonian is 
\begin{align}
\hat H_{c_1, c_2}(r_1,  p_{r_1}, r_2, p_{r_2}) &= 
\frac{n}{2}\left( p_{r_1}^2 +  \frac{c_1^2}{r_1^2}  + p_{r_2}^2 +  \frac{c_2^2}{r_2^2} \right)
+ \frac{n}{2} \sum \limits_{1\leq k < n}
V \left( 
2\sqrt{r_1^2 \sin^2 \frac{ka_1\pi}{b_1} +  r_2^2 \sin^2 \frac{ka_2 \pi}{b_2}}
% r_1^2 + r_2^2} \, \left| \sin \frac{k \pi}{n}\right|
 \right).
\label{H_red}
\end{align}
%
%
%(In the planar case, the second term
%reduces to $V \left( 
%2\sqrt{
% r_1^2 + r_2^2} \left|\sin \frac{k \pi}{n}\right|
%\right).$

If $V$ is attractive (and so $V'(r)>0$ for all $r>0$) then regular $n$-gon RE 
configurations arise as solutions of a system of two equations in $r_1, r_2$ of the following form, 
for $j=1$ and $j=2$:
\begin{align}
&\frac{c_j^2}{r_j^4} =  \sum \limits_{1\leq k< n}\frac{
V' \left( 
2\sqrt{r_1^2 \sin^2 \frac{ka_1\pi}{b_1} +  r_2^2 \sin^2 \frac{ka_2 \pi}{b_2}}
 \right)
 }
 {
 \sqrt{r_1^2 \sin^2 \frac{ka_1\pi}{b_1} +  r_2^2 \sin^2 \frac{ka_2 \pi}{b_2}}
 }  \sin^2 \frac{k a_j\pi}{b_j}.
 \label{RE_H_homo}
 \end{align}
In the \textit{planar} case, the condition is:
\begin{align}
&\frac{c_j^2}{r_j^4} =  \sum \limits_{1\leq k< n}\frac{
V' \left( 
2\sqrt{r_1^2 \sin^2 \frac{k\pi}{n} +  r_2^2 \sin^2 \frac{k \pi}{n}}
 \right)
 }
 {
 \sqrt{r_1^2 \sin^2 \frac{k\pi}{n} +  r_2^2 \sin^2 \frac{k \pi}{n}}
 }  \sin^2 \frac{k \pi}{n}.
 \label{RE_H_homo_planar}
 \end{align}
In any case (planar or non-), we observe that for at least some $(c_1, c_2)$ with $c_1\neq 0$ and $c_2 \neq 0$, the systems above admits solutions $r_{10} =r_{10}(c_1, c_2),$ $r_{20}= r_{20}(c_1, c_2)$ 
and these satisfy
\begin{align}
r_{20}= \gamma\, r_{10} \,\,\,\,\,\,\text{where}\,\,\, \,\gamma := \sqrt{ \left| \frac{c_2}{c_1} \right|}\,\,. 
\label{homog_RE}
\end{align}
The dynamics on the invariant manifolds of synchronised homographic motions coincide to the central force problem on $\mathbb{R}^4$ as studied in Section \ref{sect: central}, with the reduced Hamiltonian  given in \eqref{H_red}.

\section{The three-body problem}\label{sect:3-body}

\subsection{Reduction}

In this section we reduce the 12 degrees of freedom   the three body problem in $\mathbb{R}^4$ to a 6 degrees of freedom system. This is possible due to translational and rotational $SO(2)_{xy}\times SO(2)_{zw}$-symmetries of the dynamics.
 %the linear and angular momenta conservation, where for the latter we consider the $SO(2)_{xy}\times SO(2)_{zw}$ action. 
 Rather than applying the general symplectic reduction theory, we chose to deduce the reduced space and dynamics working specifically on our system.

\smallskip
Consider the three body problem in $\mathbb{R}^4$ with the  configuration given by $(\q_1, \q_2, \q_3) \in \left(\mathbb{R}^{4} \right)^3 \setminus \{\text{possible collisions}\}.$ We start by introducing  the Jacobi coordinates
\begin{equation}
\u:= \q_2-\q_1, \,\,\,\, \v:=\q_3 - \frac{m_1 \q_1+m_2\q_2}{m_1+m_2}
\end{equation}
that we write in double-polar coordinates 
\begin{equation}
\u:= (R_1, \Theta_{R1}, R_2, \Theta_{R2}),\, \,\,\,\, \v:= (S_1, \Theta_{S1}, S_2, \Theta_{S2}).
\end{equation}
The Hamiltonian \eqref{Ham_gen} reads:
\begin{align}
H&=\frac{1}{2M_1}\left(P_{R_1}^2 +\frac{P_{\Theta_{R1}}^2}{R_1^2} +  P_{R_2}^2 +\frac{P_{\Theta_{R2}}^2}{R_2^2}  \right)  +  
\frac{1}{2M_2}\left(P_{S_1}^2 +\frac{P_{\Theta_{S1}}^2}{S_1^2} +  P_{S_2}^2 +\frac{P_{\Theta_{S2}}^2}{S_2^2}  \right)  \nonumber \\
&\hspace{8cm} - V(d_{12}) -  V(d_{13}) -V(d_{23}) 
\label{Ham_5}
\end{align}
where $\displaystyle{M_1:={m_1m_2}/{(m_1+m_2)}}$ and $\displaystyle{M_2=m_3(m_1+m_2)/(m_1+m_2+m_3)}$, and 
\begin{align*}
&d_{12}= \sqrt{R_1^2+R_2^2}   \\
 &d_{13}= \sqrt{S_1^2+\alpha_1^2 R_1^2 - 2 \alpha_1 R_1 S_1 \cos (\Theta_{S1}-\Theta_{R1})   +S_2^2+\alpha_1^2R_2^2 - 2 \alpha_1 R_2 S_2 \cos (\Theta_{S2}-\Theta_{R2})}  \\
  &d_{23}= \sqrt{S_1^2+\alpha_2^2 R_1^2 +2 \alpha_2 R_1 S_1 \cos (\Theta_{S1}-\Theta_{R1})   +S_2^2+\alpha_2^2R_2^2 + 2 \alpha_2 R_2 S_2 \cos (\Theta_{S2}-\Theta_{R2}) }
\end{align*}
with $\alpha_1=m_2/(m_1+m_2)$ and $\alpha_2=m_1/(m_1+m_2)$. The symmetry becomes obvious after  performing the (symplectic) change of variables 
\begin{align}
&(\Theta_{Rj}, \Theta_{Sj}) \to (\Phi_j, \Psi_j) := (\Theta_{Sj} - \Theta_{Rj}, \Theta_{Sj}), \\
&(P_{\Theta_{Rj}}, P_{\Theta_{Sj}}) \to (P_{\Phi_j}, P_{\Psi_j}) := (P_{\Theta_{Rj}},  P_{\Theta_{Rj}}+ P_{\Theta_{Sj}}),\,\, \,\,\,j=1,2.
\end{align}
The Hamiltonian \eqref{Ham_5} becomes:
\begin{align}
H&=\frac{1}{2M_1}\left(P_{R_1}^2 +\frac{P_{\Phi_{1}}^2}{R_1^2} +  P_{R_2}^2 +\frac{P_{\Phi_{2}}^2}{R_2^2}  \right)  +  
\frac{1}{2M_2}\left(P_{S_1}^2 +\frac{  (P_{\Psi_{1}} - P_{\Phi_{1}}  )^2}{S_1^2} +  P_{S_2}^2 +\frac{ (P_{\Psi_{2}} - P_{\Phi_{2}}  )^2}{S_2^2}  \right)  \nonumber \\
&\hspace{7cm} - V(d_{12}) -  V(d_{13}) -V(d_{23}) 
\end{align}
with
\begin{align}
&d_{12}=\sqrt{R_1^2+R_2^2} \\
 &d_{13}= \sqrt{S_1^2+\alpha_1^2 R_1^2 - 2 \alpha_1 R_1 S_1 \cos \Phi_1   +S_2^2+\alpha_1^2R_2^2 - 2 \alpha_1 R_2 S_2 \cos \Phi_2}  \\
  &d_{23}= \sqrt{S_1^2+\alpha_2^2 R_1^2 +2 \alpha_2 R_1 S_1 \cos \Phi_1   +S_2^2+\alpha_2^2R_2^2 + 2 \alpha_2 R_2 S_2 \cos \Phi_2 }\,.
\end{align}
It is immediate that along any solution the angular momentum  is conserved and so:
\begin{equation}
\left(P_{\Psi_{1}}(t),P_{\Psi_{2}}(t) \right)= const.= \left(\mu_1, \mu_2\right)\,.
\end{equation}
Thus, denoting the reduced configurations space $M:= [0, \infty) \times (0, \infty)  \times (0, \infty)  \times (0, \infty)   \times S^1 \times S^1 \setminus \{\text{possible collisions}\}$, we obtain  the reduced Hamiltonian:
\begin{align}
H_{\text{equil}}&: T^*M \to \mathbb{R} \nonumber\\
H_{\text{equil}}&=\frac{1}{2M_1}\left(P_{R_1}^2 +\frac{P_{\Phi_1}^2}{R_1^2} +  P_{R_2}^2 +\frac{P_{\Phi_2}^2}{R_2^2}  \right)  +  
\frac{1}{2M_2}\left(P_{S_1}^2 +\frac{  (\mu_1-P_{\Phi_{1}} )^2}{S_1^2} +  P_{S_2}^2 +\frac{ (\mu_2-P_{\Phi_{2}})^2}{S_2^2}  \right)  \nonumber \\
&\hspace{6cm} - V(d_{12}) -  V(d_{13}) -V(d_{23}). 
\label{red_3b}
\end{align}
The equations of motion are
\begin{align}
&\dot R_j = \frac{P_{R_j}}{M_1} 
\label{sys_1}
\\
&\dot P_{R_j} =  \frac{P_{\Phi_j}^2}{M_1R_j^3} - \frac{V'(d_{12})}{d_{12}}R_j -  \frac{V'(d_{13})}{d_{13}}\left(\alpha_1^2 R_j- \alpha_1  S_j \cos \Phi_j \right) -   \frac{V'(d_{23})}{d_{13}}\left(\alpha_2^2 R_j+ \alpha_2  S_j \cos \Phi_j \right)\\
& \dot S_j= \frac{P_{S_j}}{M_2}\\
& \dot P_{S_j} =  \frac{(\mu_j -P_{\Phi_{j}})^2}{M_2S_j^3} - \frac{V'(d_{13})}{d_{13}}\left(S_j - \alpha_1  R_j \cos \Phi_j \right) -   \frac{V'(d_{23})}{d_{13}}\left(S_j+ \alpha_2  R_j \cos \Phi_j \right)\\
&\dot \Phi_j = \frac{P_{\Phi_j} }{M_1 R_j^2} -  \frac{\mu_j -P_{\Phi_{j}}}{M_2 S_j^2}\\
& \dot P_{\Phi_j}= -  \frac{V'(d_{13})}{d_{13}}\left(\alpha_1 R_j S_j \sin \Phi_j \right)
+  \frac{V'(d_{23})}{d_{23}}\left(\alpha_2 R_j S_j \sin \Phi_j \right).
\label{sys_6}
\end{align}

\subsection{Relative equilibria}

Recall that  in a Lie-symmetric   mechanical system,  RE solutions project to equilibria in a 
reduced space, and vice-versa,
any equilibrium  in a reduced space lifts to a RE solution in the unreduced space (for details, see \cite{Ma92}).

In the case of the three body problem  previously presented, for every $\muu=\left(\mu_1, \mu_2\right) \setminus (0,0)$,  the RE  (of momentum $\muu$) are found as the equilibria of the system \eqref{sys_1}-\eqref{sys_6}. We observe that  equating to zero the RHS of the \eqref{sys_6} we obtain
\begin{align}
R_j S_j \sin \Phi_j  \left( \alpha_1  \frac{V'(d_{13})}{d_{13}} - \alpha_2  \frac{V'(d_{23})}{d_{23}}\right)=0, \,\,\,j=1,2
\end{align}
and so
\begin{itemize}
\item either $\sin \Phi_j =0$ i.e. the projections of $\u$ and $\v$ are parallel on the principal planes, i.e.
\[
\text{Proj}_{xy}( \u )\| \text{Proj}_{xy} (\v)\,\,\,\text{and}\,\,\,\text{Proj}_{zw}( \u )\| \text{Proj}_{zw} (\v)
\]
These RE correspond to the collinear Euler configurations in the classical thee body problem;

\item or 
\begin{equation}
m_2  \frac{V'(d_{13})}{d_{13}} = m_1  \frac{V'(d_{23})}{d_{23}},
\end{equation}
where we use that $\alpha_1=m_2/(m_1+m_2)$ and $\alpha_2=m_1/(m_1+m_2).$
\end{itemize}
\begin{remark}
If we chose $V$ to be the Newtonian potential, i.e. for any two mass points $m_j, m_k$ we have $V(d_{ij}) =  -m_j m_k/d_{jk}$ then the condition above leads to $d_{13}=d_{23}$, and so any non-collinear three body  RE triangle is isosceles. This was also proved (with a different method)  in the general case of the $n$-body problem in $\mathbb{R}^n$, n even, by  Albouy and Chenciner \cite{AC98}.

\end{remark}

\subsection{Stability of equilateral triangles}

In this Section we study the stability of the equilateral triangle solutions in the case of three equal masses interacting via a generic  attractive potential.
%, exception being the Jacobi $v(r)=-1/r^2$ potential. 
%Note that for $n=3$ the synchronised and non-synchronised homographic motions coincide.
%From Section \ref{sect:DR},  the side of the equilateral triangles formed by the three bodies is 
%
%\begin{align}
%l=r_0\sqrt{3\left(1+ \gamma^2 \right)}\,.
%\end{align}
%
%where $r_0$ is the 

%Consider now that $n=3$ and  $m_1=m_2=m_3=1.$
 %Assuming that $V$ is attractive, 
 From Section \ref{sect:DR}, equations 
\eqref{RE_H_homo_planar} and \eqref{homog_RE}, given $(c_1, c_2) \in  \in \mathbb{R}^2\setminus \{(0,0)\}$
 the RE configuration polar radii receive the form
%%%%%%%%%%%%%%%%%%%%%%%%%%%%%%%%%%%%%%%%%%%
%
%$\displaystyle{
%r_2 = r_1 \sqrt{\left| \frac{c_2}{c_1} \right|}}$.  Denoting $\displaystyle{\gamma:= \sqrt{\left| \frac{c_2}{c_1}\right| } = \sqrt{\left| \frac{\mu_2}{\mu_1}\right| } }$  (where we used  that $3c_j=\mu_j,$ $j=1,2$) we may write
%
\begin{align}
r_1=: r_0, \,\,\,\,\, r_2 = \gamma r_0.
\end{align}
where $r_0$ solves
\begin{align}
 r^3V'\left(r\sqrt{3\left(1+ \gamma^2 \right)} \right)=\frac{2\sqrt{3\left(1+\gamma^2\right)}}{9} \,c_1^2
\label{r_0}
\end{align}
where
\[
 \gamma :=  
 \sqrt{\left|  \frac{\mu_2}{\mu_1}\right| }>0.
\]
Note that  $\mu_j=3c_j,$ $j=1,2$.
% =\sqrt{\left| \frac{c_2}{c_1}\right|}  >0.

\medskip
To ease notation,   we write  $\mu:=\mu_1 =  3c_1$. We assume that we are in a generic situation in which for any $(\mu, \gamma)$ in some non-void domain, the equation \eqref{r_0} has at least one 
solution $r_0 = r_0(\mu, \gamma)$. 
The RE with the polar radii above project into the equilibria  
\begin{align}
&(R_1, R_2, S_1, S_2, \Phi_1, \Phi_2) = \left( r_0\sqrt{3},\, \gamma r_0\sqrt{3},\, \frac{3r_0}{2},\, \frac{3\,\gamma\, r_0}{2}, \,-\frac{\pi}{2}, -\frac{\pi}{2} 
\right)
\label{equi_1}
\\
&(P_{R_1}, P_{R_2}, P_{S_1}, P_{S_2}, P_{\Phi_1}, P_{\Phi_2}) =  \left(0, 0, 0, 0, \frac{\mu}{2}, \frac{\gamma^2\mu}{2} 
\right).
\label{equi_2}
\end{align}
of the system \eqref{sys_1} - \eqref{sys_6}. We also note that the side of the equilateral triangles formed by the three bodies is 
\begin{align}
l=r_0\sqrt{3\left(1+ \gamma^2 \right)}\,.
\end{align}

We will show, in Theorem \ref{unstable} below, that for generic potentials and generic values of $\mu$ and $\gamma$, 
the equilateral configuration RE is unstable.

%Applying the general theory for finite dimensional Hamiltonian systems  (see for instance, \cite{Pat92, Ma92}) to our context, a RE solution of  is Lyapunov/linearly/spectrally stable modulo the translational and rotational $SO(2)_{xy}\times SO(2)_{zw}$ symmetries if the corresponding equilibrium solution   of the reduced reduced  system  is Lyapunov/linearly/spectrally stable in the reduced space. 

Recall that an equilibrium $\z_e$ of a  Hamiltonian system with Hamiltonian $H$ is: \textit{Lyapunov (or non-linearly) stable} if the Hessian $D^2H(\z_e)$ is positive definite;  \textit{linearly stable} if the linearization matrix $\mathbb{J}D^2H(\z_e)$, where $\displaystyle{\mathbb{J} = \left[
\begin{array}{cccc}
\,\,\,\,{\mathbb{O}}_6 & \mathbb{I}_6 \\
-\mathbb{I}_6 & {\mathbb{O}}_6
\end{array}
\right]
 }$, is semi-simple (diagonalizable) and has all eigenvalues purely imaginary; and \textit{spectrally stable} if none of its eigenvalues has a positive real part \cite{MR94}. 
Since the eigenvalues of a Hamiltonian system always appear in quadruples of the form $\pm \text{Re} \lambda \pm i\, \text{Im} \lambda$ (or pairs, in the case of real or purely imaginary values) \cite{MHO09}, 
spectral stability is only possible if all eigenvalues are pure imaginary.
Even in this case, the equilibrium is unstable if the semi-simple-nilpotent decomposition of the linearization has a non-trivial nilpotent component \cite{Me07}. 
For symplectic Hamiltonian systems, since the phase space is even dimensional,
if there is a zero eigenvalue, then its algebraic multiplicity must be even
(since all other eigenvalues must come in quadruples or pairs).
 % In particular for Hamiltonian systems, since the eigenvalues appear quadruples of the form $\pm \text{Re} \lambda \pm i\, \text{Im} \lambda$,  if the semi-simple-nilpotent decomposition of  the linearization possesses a non-trivial nilpotent component, then the equilibrium is unstable.

\smallskip
At an equilibrium $\z_e =  (R_1^e, R_2^e,S_1^e, S_2^e, \Phi_1^e, \Phi_1^e, 0, 0,0, 0, P_{\Phi_1}^e, P_{\Phi_2}^e) \in T^*M$ of the reduced system  \eqref{sys_1} - \eqref{sys_6}, the Hamiltonian Hessian $D^2H_{\text{equil}}(\z_e)$ is a $12 \times 12$ matrix arranged into four $6 \times 6$ blocks:
\begin{align}
D^2H_{\text{equil}}(\z_e) = 
\left[
\begin{array}{cccc}
A &B \\
B^t & D
\end{array}
\right]
\end{align}
with:
\begin{align}
A =  
\left[
\begin{array}{ccc}
[a_1] &{\mathbb{O}}_2 & {\mathbb{O}}_2\\
{\mathbb{O}}_2& [a_2] & {\mathbb{O}}_2\\
{\mathbb{O}}_2&\mathbb{O}_2&{\mathbb{O}}_2
\end{array}
\right] +  \left[D^2 V\right]\Big|_{\z=\z_e}\, ,
\label{block_A}
\end{align}
\begin{align}
[a_1] = 
\left[
\begin{array}{cc}
\frac{3P_{\Phi_1}^2 }{M_1 R_1^4}&0\\
0&\frac{3P_{\Phi_2}^2 }{M_1 R_2^4}
\end{array}
\right],\,\,\,
[a_2]=
\left[
\begin{array}{cc}
\frac{\left( \mu_1-  P_{\Phi_1} \right)^2}{M_2 S_1^4}&0\\
0&\frac{\left( \mu_2-  P_{\Phi_2} \right)^2 }{M_2 S_2^4}
\end{array}
\right],
\end{align}
\begin{align}
B =  
\left[
\begin{array}{ccc}
 {\mathbb{O}}_2 &{\mathbb{O}}_2 &[b_1]\\
{\mathbb{O}}_2 &{\mathbb{O}}_2&  [b_2] \\
{\mathbb{O}}_2&\mathbb{O}_2&\mathbb{O}_2
\end{array}
\right],
\end{align}
\begin{align}
[b_1] = 
\left[
\begin{array}{cc}
-\frac{P_{\Phi_1} }{M_1 R_1^3}&0\\
0& - \frac{P_{\Phi_2} }{M_1 R_2^3}
\end{array}
\right],\,\,\,
[b_2]=
\left[
\begin{array}{cc}
\frac{\left( \mu_1-  P_{\Phi_1} \right)}{M_2 S_1^3}&0\\
0&\frac{\left( \mu_2-  P_{\Phi_2} \right)}{M_2 S_2^3}
\end{array}
\right],
\end{align}
\begin{align}
D =  
\left[
\begin{array}{ccc}
[d_1] & {\mathbb{O}}_2 &{\mathbb{O}}_2\\
{\mathbb{O}}_2&  [d_2] &{\mathbb{O}}_2 \\
{\mathbb{O}}_2&\mathbb{O}_2&[d_3]
\end{array}
\right],
\end{align}
\begin{align}
[d_1]=\text{diag}\left(\frac{1}{M_1} \right),  \,\,\,\,[d_2]=\text{diag} \left(\frac{1}{M_2}\right), 
\end{align}
\begin{align}
[d_3]=d \left[
\begin{array}{cc}
1 & 0 \\
0&  \gamma  
\end{array}
\right], \,\,\,\,\text{where}\,\,\,\,d=\frac{1}{M_1} \text{diag} \left(\frac{1}{R_1}\right) +  \frac{1}{M_2}\text{diag}\left(\frac{1}{S_1}\right).
\end{align}
%
%FALSE It follows that $D^2H_{\text{equil}}(\z_e)$ is positive definite if the block $A$ is positive definite

\iffalse
\begin{align}
\mathbb{J} D^2H_{\text{equil}}(\z_e) = 
\left[
\begin{array}{cccc}
B^t &D \\
-A& -B
\end{array}
\right].
\end{align}
\fi

\smallskip
We now consider the equilateral triangular RE  given by formulae \eqref{equi_1}-\eqref{equi_2}. In this case the Hessian matrix $D^2V\big|_{\z=\z_e}$ takes the form
\begin{align}\label{E:DV2}
D^2V\big|_{\z=\z_e}
\left[
\begin{array}{cccccc}
*&*&*&*&0&0\\
*&*&*&*&0&0\\
*&*&*&*&0&0\\
*&*&*&*&0&0\\
0&0&0&0&a&\gamma^2 a \\
0&0&0&0&\gamma^2 a&\gamma^4a\\
\end{array}
\right]
\end{align}
where $*$ are the appropriate partial derivatives evaluated at the RE  $\partial V_{ij}\big|_{\z=\z_e}$, $1\leq i\leq j\leq 4$ entries, and
\begin{equation}
a:=\frac{3\sqrt{3}}{8(1+\gamma^2)\sqrt{1+\gamma^2}}\left( l V''(l)- V'(l) \right).
\end{equation}
Let us denote
\[
[a]:=
\left[
\begin{array}{cc}
a&\gamma^2 a \\
\gamma^2 a&\gamma^4a\\
\end{array}
\right]
=a \left[
\begin{array}{cc}
1&\gamma^2 \\
\gamma^2 &\gamma^4\\
\end{array}
\right].
\]

%
%Assuming that we are in the generic case of $l V''(l)- V'(l)\neq 0$, we notice that 
The eigenvalues $\lambda_1, \lambda_2$  of $[a]$, 
and corresponding eigenvectors $\u_1$ and $\u_2$, are
\begin{align}
&\lambda_1= 0, &&\u_1=(-\gamma^2, 1),
\label{zero_eigenval}
\\
&\lambda_2 = a\left(\gamma^4+1\right), 
%&\lambda_2=\frac{ 
%3l(1+\gamma^4)
%\left(lV''(l)- 
%V'(l)
%\right)
%}{8(1+\gamma^2)^2},
&&\u_2=\left( 1,\gamma^2 \right)\,. \notag
\end{align}
Since $[a]$ has a zero eigenvalue, it follows from the structure given above 
(including Equation \ref{E:DV2}) 
that $D^2H_{\text{equil}}(\z_e)$ also has a zero eigenvalue.
The corresponding eigenspace is
\begin{align*}
\ker D^2H_{\text{equil}}(\z_e)
= \left\{ (\v_1,\dots,\v_6) \in \R^{12} \,| \, \v_3 \in \ker [a], \v_4=\v_5=0 \textrm{ and }
(\v_1, \v_2, \v_6) \in \ker [F]
\right\},
\end{align*}
where
\begin{align*}
[F] = \left[
\begin{array}{ccc}
* & * & [b_1] \\
* & * & [b_2]\\
{[ }b_1] & [b_2] & [d_3]
\end{array}
\right],
\end{align*}
in which the upper left-hand $2\times 2$ block equals that of $A$.
Note that for generic potentials $V$, and generic values of $\mu$ and $\gamma$, 
the matrix $[F]$ is nonsingular, and
hence $\ker D^2H_{\text{equil}}(\z_e)$ has dimension $1$.
%In particular, $\det [a] = 0$.
%From the block structure of $D^2H_{\text{equil}}(\z_e)$ given above,
%$\det D^2H_{\text{equil}}(\z_e)$ is a multiple of $\det [a]$, and so it also equals $0$.
%Thus $D^2H_{\text{equil}}(\z_e)$ has a zero eigen
%The geometric multipli

Since $\mathbb{J}$ is invertible, it follows that the linearization 
$\mathbb{J}D^2H_{\text{equil}}(\z_e)$ also has a kernel of  dimension at least $1$ 
%(or, equivalently, it has a zero eigenvalue with geometric multiplicity $1$). 
Under the condition det$[F] \neq 0$, the kernel  of $\mathbb{J}D^2H_{\text{equil}}(\z_e)$  has dimension exactly 1. 
The infinitesimally symplectic structure of $\mathbb{J}D^2H_{\text{equil}}(\z_e)$
implies that the \textit{algebraic} multiplicity of the zero eigenvalue is at least two. Thus, if det$[F] \neq 0$, the Jordan normal 
form of $\mathbb{J} D^2H_{\text{equil}}(\z_e)$ 
contains a block of the form%
\[
\left[
\begin{array}{cc}
0& 1 \\
0&0\\
\end{array}
\right]
\]
or, equivalently, the semi-simple-nilpotent decomposition has a nontrivial nilpotent term.

Finally, in a spectrally stable system, the presence of a nontrivial nilpotent term
induces a ``drift'' in the dynamics, implying (nonlinear) instability (see, for instance, \cite{Me07}).
Thus we have proven the following,

\begin{theorem}\label{unstable}
For generic potentials, the equilateral configuration RE is unstable.
\end{theorem}

\bigskip
For example, for the classical Newtonian potential $V(r)=-1/r$, numerical evidence shows that the det$[F]$ is strictly negative for all permitted values of $\mu$ and $\gamma$ and so the ``Newtonian" equilateral configuration RE is unstable.
%As a final observation, we note that on the case of Jacobi potential $V(r)=-1/r^2,$ the equation about has a continuum of solutions and thus it will be studied separately elsewhere.)
The same is valid for the Jacobi potential $V(r)=-1/r^2$,  with the observation that in this case any $r>0$ is  a solution of \eqref{r_0} (i.e. there is a continuum of solutions) as long as $\mu$ and $\gamma$ lie along the curves given by $\displaystyle{\mu^2= 3/\left(2(1+\gamma^2)^2\right).}$

\section{Final remarks}

The present study  is the tip of the iceberg of interesting research. One may ask a multitude of  questions: is there an equivalent of Bertrand's theorem in $\mathbb{R}^4$? for what potentials can we regularize  collisions? what are the equivalent of the figure 8 solutions (see \cite{CM00})? what is the equivalent of Saari's conjecture (see \cite{Sa05}) and it is true in $\mathbb{R}^4$? how many REs can be found for a  fixed $n\geq 3$ and what is their stability? etc. Also, an interesting study may be made of the $\mathbb{R}^4$ gravitational potential, that is the Jacobi potential $V(r)=-k/r^2$, $k>0$. It is also worth mentioning that that the latter marks the threshold between weak and strong forces as defined in physics. 
Last but not least, one may also study the generalization of  the $n$-body problem in  $\mathbb{R}^{n}$, $n\geq 4$, continuing the research started here and elsewhere.  We leave all these for future endeavours.

\section*{Acknowledgement}
CS thanks University of Ottawa for hosting her during her Sabbatical. The authors thank the anonymous reviewers for useful comments and remarks.

\enlargethispage{20pt}

%\ethics{Insert ethics statement here if applicable.}

%\dataccess{Insert details of how to access any supporting data here.}

%\aucontribute{CS is the lead author of this work. TS contributed to the geometric
%setting and study of relative equilibria.}

%\disclaimer{Insert disclaimer text here if applicable.}

%%%%%%%%%% Insert bibliography here %%%%%%%%%%%%%%


\begin{thebibliography}{99}



%%%%%%%%%%%%%%%%%%%%%

\bibitem[AB78]{AB78}   Abraham R.,  Marsden J.E. 1978. \textit{Foundations of Mechanics}.  Benjamin/Cummings Publishing Co., Inc., Advanced Book Program, Reading, Mass.

\bibitem [AC98]{AC98}  Albouy A., Chenciner A. 1998. Le probl\'eme des n corps et les distances mutuelles. 
\textit{Inventiones Mathematicæ} \textbf{131} 


\bibitem[Al15]{Al15} Albouy A. 2015 On the Force Fields Which Are Homogeneous of Degree ?3. In: Corbera M., Cors J., Llibre J., Korobeinikov A. (eds) \textit{Extended Abstracts Spring 2014}. Trends in Mathematics, vol 4. Birkhäuser, Cham



\bibitem[AMP16]{AMP16}  Anco S.C., Meadows T. and   Pascuzzi V.: 2016. Some new aspects of first integrals and symmetries for central force dynamics. \textit{Journal of Mathematical Physics} \textbf{57}

\bibitem[APS14]{APS14} Arredondo J.A., Perez-Chavela E. and Stoica C. 2014. Dynamics in the Schwarzschild isosceles three body problem.  \textit{Journal of Nonlinear Sciences}  \textbf{24}







\bibitem[BCV17]{BVC17}Barrabes E., Cors J.M., and Vidal C. 2017. Spatial collinear restricted four-body problem with repulsive Manev potential. \textit{Celest. Mech. Dyn. Astron.} \textbf{129}

%\bibitem[CT11]{CT11} Castelli R. and Terracini S. 2011. On the regularization of the collision solutions of the one-center problem with weak forces. \textit{Discrete Contin. Dyn. Syst.} \textbf{31} 


\bibitem[CM00]{CM00} Chenciner A. and R. Montgomery R. 2000. A remarkable periodic solution of the three body problem in the case of equal masses, \textit{Annals of Mathematics} \textbf{152}


\bibitem[Ch13]{Ch13} Chenciner A. 2013. The Lagrange reduction of the $n$-body problem. A survey. \textit{Acta Math Vietnam}
\textbf{38}

\bibitem[Cox36]{Cox36} Coxeter H. S. M. 1938. Regular Skew Polyhedra in Three and Four Dimension, and their Topological Analogues, 
\textit{Proceedings of the London Mathematical Society}, \textbf{s2-43}:1, 

\bibitem[DMS00]{DMS00} Diacu F.,  Mioc V. and Stoica C. 2000. Phase-Space Structure and Regularization of Manev-Type Problems. \textit{Nonlinear Analysis} \textbf{41}



%\bibitem[AK12]{AK12} Albouy A, Kaloshin V. 2012. Finiteness of central configurations of five bodies in the plane  \textit[Annals of Mathematics} \textbf{176}

\bibitem[Dia14]{Dia14}Diacu F. 2014. Relative equilibria in the 3-dimensional curved $n$-body problem. \textit{Mem. Amer. Math. Soc.}\textbf{228}

\bibitem[DSS18]{DS18} Diacu F., Stoica C. and Zhu S. 2018. Central configurations of the curved $n$-body problem. \textit{J. Nonlinear Sci.} \textbf{28}

\bibitem[GS84]{GS84} Guillemin V. and Sternberg S. 1984. \textit{Symplectic Techniques in Physics}, Cambridge Univ. Press


\bibitem [Go77]{Go77} Gordon WB. 1977. A Minimizing Property of Keplerian Orbits.
\textit{American Journal of Mathematics} \textbf{99}



%\bibitem[HSS09]{HSS09} Holm D.D, Schmah T., Stoica  C. 2009. \textit{Geometry, Symmetry and Mechanics: from finite to infinite-dimensions}. Oxford Texts in Applied and Engineering Mathematics, \textbf{12}, Oxford University  Press

\bibitem[Hall15]{Hall15} Hall B. 2015. \textit{Lie groups, Lie algebras, and representations: an elementary introduction} \textbf{222}, Springer.




\bibitem[Ma92]{Ma92} Marsden J.E. 1992. \textit{Lectures on Mechanics}. 
 London Math.\ Soc.\ Lecture Note Ser., \textbf{174}.
  Cambridge University Press


\bibitem[Me07]{Me07} Meiss J. 2007. \textit{Differential dynamical systems}. SIAM


\bibitem[MG81]{MG81}McGehee R. 1981. Double collisions for a classical particle system with non-gravitational interactions. \textit{Comment. Math. Helv.} \textbf{56} 


\bibitem[Mo14]{Mo14} Moeckel R. 2014. 
\textit{Lectures on Central Configurations}. Notes of the  Centre de Recerca Matematica, Barcelona





\bibitem[MR94]{MR94} Marsden J.E.,  Ratiu T.S. 1994. \textit{Introduction to mechanics and symmetry. A basic exposition of classical mechanical systems}.  Texts in Applied Mathematics, \textbf{17}.  Springer-Verlag, New York 

\bibitem[MHO09]{MHO09} Meyer K.  Hall G. and Offin D. 2009. \textit{Introduction to Hamiltonian Systems and the $n$-body Problem.}
Applied Mathematical Sciences, \textbf{90}, 
 2$^{nd}$ edition,  Springer-Verlag, New York





\bibitem[OV06]{OV06} Onder F., Vercin A. 2006. Orbits of the n-dimensional Kepler-Coulomb problem and universality of the Kepler laws. \textit{Eur. J. Phys}  \textbf{27}

\bibitem[Pal79]{Pal79}Palais R. 1979. The principle of symmetric criticality. \textit{Comm. Math. Phys.} \textbf{69}




\bibitem[PJ80]{PJ80} Palmore J. 1980. Relative equilibria of the $n$-body problem in $E^4$. \textit{Journal of Differential equations} \textbf{38}




\bibitem[PJ81a]{PJ81a} Palmore J. 1981. 
Homographic solutions of the $n$-body problem in $E^4$.
\textit{Journal of Differential equations}   \textbf{40}

\bibitem[PJ81b]{PJ81b} Palmore J. 1981. 
Central configurations of the restricted problem in $E^4$. 
\textit{Journal of Differential equations}   \textbf{40}


\bibitem[Pat92]{Pat92}Patrick G.W. 1992. 
Relative equilibria in Hamiltonian systems: the dynamic interpretation of nonlinear stability on a reduced phase space. 
\textit{Journal of Geometry and Physics} \textbf{9}

\bibitem[MS15]{MS15} Montgomery R. and Shanbrom C. 2015. Keplerian motion on the Heisenberg group and elsewhere, in: D.E. Chang, D.D.Holm, G. Patrick, T. Ratiu (Eds.), \textit{Geometric Mechanics: The Legacy of Jerry Marsden, in: Fields Institute Communications Series},  Springer-Verlag, New York

\bibitem[Sa05]{Sa05}Saari D. 2005. \textit{Collisions, Rings, and Other Newtonian N-Body Problems}, American Math Society, Providence RI


\bibitem[Shc06]{Shc06} Shchepetilov A.V. 2006. \textit{Calculus and Mechanics on Two-Point Homogeneous Riemannian Spaces} Lecture Notes in Physics \textbf{707}


\bibitem[Sti08]{Still08} Stillwell J. 2008. \textit{Naive lie theory}, Springer-Verlag, New York



\bibitem[St00]{St00}Stoica C. 2000. Particle systems with quasihomogeneous interaction PhD Thesis, University of Victoria

\bibitem[St18]{St18} Stoica C. 2018. On the $n$-body problem on surfaces of revolution. \textit{Journal of Differential Equations}  
 \textbf{264}




\end{thebibliography}
\end{document}